\documentclass{article}
\usepackage[a4paper, margin=2.5cm]{geometry}
\usepackage[T1]{fontenc}
\usepackage{amsmath, amssymb, amsthm}
\usepackage{xcolor, graphicx}
\usepackage{tikz, tikz-cd}
\usepackage{fancyhdr}
\usepackage{hyperref}
\usepackage{enumitem}

\usepackage[normalem]{ulem}

\hypersetup{
    colorlinks = true,
    linkcolor = red,
    anchorcolor = black,
    citecolor = blue,
    filecolor = cyan,
    menucolor = red,
    runcolor = cyan,
    urlcolor = blue,
    linkbordercolor = {white},
    linktocpage = true
}

{\end{list}}%

\newtheorem{theorem}{Theorem}[section]
\newtheorem{proposition}[theorem]{Proposition}
\newtheorem{lemma}[theorem]{Lemma}
\newtheorem{corollary}[theorem]{Corollary}
\newtheorem{hypotheses}[theorem]{Hypotheses}

\theoremstyle{definition}
\newtheorem{example}[theorem]{Example}
\newtheorem{remark}[theorem]{Remark}

\makeatletter
\def\thmhead@plain#1#2#3{%
  \thmname{#1}\thmnumber{\@ifnotempty{#1}{ }\@upn{#2}}%
  \thmnote{ {\the\thm@notefont#3}}}
\let\thmhead\thmhead@plain
\makeatother

\newcommand{\F}{\mathbb{F}}
\newcommand{\Fq}{\mathbb{F}_q}
\newcommand{\Fqq}{\mathbb{F}_{q^2}}
\newcommand{\Fqr}[1]{\mathbb{F}_{q^{#1}}}

\newcommand{\cC}{\mathcal{C}}
\newcommand{\bbC}{\mathbb{C}}

\newcommand{\car}{\mathrm{char}}
\newcommand{\mx}{\mathrm{x}}

\usepackage{array}

\begin{document}

\renewcommand{\headrulewidth}{0pt}

\rhead{ }
\lhead{ }

\title{Dualities of dihedral and generalised quaternion codes and applications to quantum codes\renewcommand\thefootnote{\arabic{footnote}}\footnotemark[1]}

\author{\renewcommand\thefootnote{\arabic{footnote}}
Miguel Sales-Cabrera\footnotemark[2], 
\renewcommand\thefootnote{\arabic{footnote}} 
 Xaro Soler-Escriv\`a\footnotemark[2],
  \renewcommand\thefootnote{\arabic{footnote}} 
  V\'ictor Sotomayor\footnotemark[3]
 }

\footnotetext[1]{This study forms part of the Quantum Communication programme and was supported by MCIN with funding from European Union NextGenerationEU (PRTR-C17.I1) and by Generalitat Valenciana COMCUANTICA/008. This work is also partially supported by the Ministerio de Ciencia e Innovaci\'on project PID2022-142159OB-I00. The second author is supported by the Generalitat Valenciana project CIAICO/2022/167. The first and second authors are supported by the I+D+i projects VIGROB23-287 and UADIF23-132 of Universitat d'Alacant.}

\footnotetext[2]{Dpt.\ de Matem\`atiques, Universitat d'Alacant, Sant Vicent del Raspeig, Ap.\ Correus 99, 03080 Alacant (Spain).} 

\footnotetext[3]{Departamento de Álgebra, Facultad de Ciencias, Universidad de Granada, Av. Fuente Nueva s/n, 18071 Granada (Spain). \\  E-mail adresses: \texttt{miguel.sales@ua.es, xaro.soler@ua.es, vsotomayor@ugr.es}}

{\small \date{\today}} 

\maketitle


\begin{abstract}
Let $\Fq$ be a finite field of $q$ elements, for some prime power $q$, and let $G$ be a finite group. A (left) \emph{group code}, or simply a \emph{$G$-code}, is a (left) ideal of the group algebra $\Fq[G]$. In this paper, we provide a complete group-algebraic description for the Hermitian dual code of any $D_n$-code over $\Fqq$, where $D_n$ is a dihedral group of order $2n$ with $n$ not divisible by $\car(\Fqq)$, through a suitable Wedderburn-Artin decomposition of the group algebra $\Fqq[D_n]$, and we determine all distinct Hermitian self-orthogonal $D_n$-codes over $\Fqq$. We also present a thorough representation of the Euclidean dual code of any $Q_n$-code over $\Fq$, where $Q_n$ is a generalised quaternion group of order $4n$ not divisible by $\car(\Fq)$, via the Wedderburn-Artin decomposition of the group algebra $\Fq[Q_n]$. In particular, since the semisimple group algebras $\Fqq[Q_n]$ and $\Fqq[D_{2n}]$ are isomorphic, then the Hermitian dual code of any $Q_n$-code has also been fully described. As an application of the Hermitian dualities computed, we give a systematic construction, via the structure of the group algebra, to obtain quantum error-correcting codes, and in fact, with this methodical approach, we recover some already known quantum codes that achieve the best known minimum distance for their length and dimension.
\end{abstract}



\section{Introduction}\label{sec:Introduction}

Let $G$ be a group of order $n$, and let $\Fq$ be a finite field of $q$ elements, for some prime power $q$. Recall that the set $\Fq[G]$ of all $\Fq$-linear combinations of the elements of $G$ forms an $\Fq$-vector space with the elements of $G$ as basis, and endowed with multiplication as follows $$\left( \sum_{g\in G} a_gg\right)\cdot \left( \sum_{g\in G} b_gg\right) := \sum_{g\in G}\left(\sum_{h\in G} a_hb_{h^{-1}g}\right)g$$ is an algebra, called the \emph{group algebra} of $G$ over $\Fq$. A linear code $\mathcal{C}\subseteq \Fq^n$ is called a (left)  \emph{group code} (or simply a \emph{$G$-code}) if there exists an isomorphism $\phi: \Fq^n\longrightarrow \Fq[G]$ such that $\phi(\mathcal{C})$ is closed under (left) multiplication by every element of $G$. In other words, a $G$-code is any (left) ideal of $\Fq[G]$. Hereafter, when we talk about ideals of the algebra $\Fq[G]$, we always assume that they are left ideals; in fact, since there is a bijection between left and right ideals of $\Fq[G]$, it is enough to focus on left ideals only.

We emphasise that there is no one-to-one correspondence between linear codes and ideals of some group algebra. First, a linear code that is a $G$-code does not determine the group $G$: for instance $\mathcal{C}=\{000000,111111\}\subseteq \F_{2}^6$ is a $G$-code for any group $G$ of order 6. Moreover, not every linear code can be seen as a $G$-code for some group $G$, and an intrinsic criterion for that fact is provided in \cite{bernal2009intrinsical}.

The starting point of this algebraic approach to study linear codes may be placed in the papers \cite{Berman} and \cite{MacWilliams}, where S. Berman and F. MacWilliams independently introduced the notion of a group code, respectively. In fact, it was discovered in \cite{Berman} that Reed-Muller codes over a finite field of characteristic $2$ may be seen as ideals in the group algebra of an elementary abelian $2$-group, and this was later extended for Generalised Reed-Muller codes (\emph{i.e.} for odd characteristic) in \cite{Charpin}. Many well-known families of linear codes have been proved to be $G$-codes for certain finite groups $G$, and we refer the interested reader on this topic to the survey \cite{KS2001}. There are several reasons to support the study of group codes. One of them is that they are asymptotically good (\emph{cf.} \cite{BorelloWillems}). Moreover, their algebraic structure can be deeply analysed using powerful tools from representation theory, some of which will be described in section~\ref{subsec:groupcodes}. An instance of this feature is given in \cite{LandrockManz}, which takes advantage of the underlying algebra structure to propose a new decoding algorithm for Reed-Muller codes. This illustrates once again that the more mathematical structure an object possesses, the better are the descriptions we can obtain.

Observe that group codes generalise cyclic codes, since any cyclic code of length $n$ can be seen as a $C_n$-code, where $C_n$ is a cyclic group of $n$ elements. When $G$ is abelian, then $G$-codes are also called abelian group codes and have been extensively studied from the inception of this theory. Nevertheless, the research on non-abelian group codes has been attracting increasing interest only in recent years (\emph{cf.} \cite{Gao2020, GaoYue21, Ricardo2023, VedDeu21}, to name a few). One of the main motivations for this is that quantum computers seem not able to break, so far, the security of public-key cryptosystems based on non-abelian group codes (see \cite{Sendrier}). 

Brochero-Martínez described in \cite{Bro15} the Wedderburn-Artin decomposition of the group algebra $\Fq[D_n]$ of any dihedral group of order $2n$ in the semisimple case, \emph{i.e.} whenever the characteristic of the field $\Fq$ does not divide $2n$. Following Brochero-Martínez's spirit, Gao and Yue computed the explicit decomposition of the semisimple group algebra $\Fq[Q_n]$ for a generalised quaternion group $Q_n$ of order $4n$. Later on, the decomposition of $\Fq[D_n]$ provided in \cite{Bro15} was utilised by Vedenev and Deundyak to obtain in \cite{VedDeu21} a detailed algebraic description of all $D_n$-codes over $\Fq$ and their Euclidean dual codes. The aim of the present paper is to complement this study: on the one hand, we use the decomposition of $\Fq[Q_n]$ to obtain the Euclidean dual code of any $Q_n$-code over $\Fq$ in the semisimple case and, on the other hand, we analyse Hermitian dual codes associated with dihedral and generalised quaternion group codes by decomposing the corresponding group algebra over $\Fqq$. It is worth noting that  Cao \emph{et al.} also determined algebraic expressions of $D_n$-codes, $Q_n$-codes, and their dual codes, but with a different approach: through their concatenated structure as quasi-cyclic codes (see \cite{CaoCaoFu16, CaoCaoFu23, CaoCaoMa22}).

The study of the dualities of a linear code is a fundamental topic within coding theory due to its significant applications. Among these, one application of particular current relevance is the construction of a particular kind of codes that correct errors occurring in quantum information processing, the so-called CSS (Calderbank-Shor-Steane) quantum codes (\emph{cf.} \cite{KKKK, Laguardia20}). This construction, in the Euclidean case, utilises two classical linear nested codes (or a self-orthogonal linear code) to address the problem of correcting phase and flip errors in a qubit. The authors of this paper applied that construction in \cite{Nostre24}  in the particular case of $D_n$-codes over $\Fq$, taking advantage of the methodical representation of their Euclidean dual codes computed in \cite{VedDeu21} in terms of the structure of $\Fq[D_n]$. There is also a CSS construction based on the Hermitian metric, which utilises a Hermitian self-orthogonal linear code, and experimental results generally yield quantum codes with better parameters since larger fields are considered. Motivated by this last fact, we will apply our algebraic characterisation of the Hermitian dual of a $D_n$-code over $\Fqq$ based on the structure of $\Fqq[D_n]$ to provide a systematic construction of CSS quantum dihedral codes, avoiding thus brute force computations (compare, for instance, with \cite{Yu24}).

Our concrete contributions to the aforementioned research line are as follows. In Section 2 we introduce the necessary background on group algebras, group codes and their dualities. After that, we propose in Section 3 a systematic group-algebraic description of the Hermitian duality of dihedral codes. More concretely, in Theorem~\ref{theo_dih_refined} we refine the Wedderburn-Artin decomposition of $\Fqq[D_n]$ given by Brochero-Martínez to obtain a group-algebraic description of the Hermitian dual code of any $D_n$-code over $\Fqq$ whenever the characteristic of the field does not divide $n$. As a consequence, we determine all Hermitian self-orthogonal $D_n$-codes in Theorem~\ref{dih-self-hermitic}. In Section 4, we use analogous group-algebraic techniques to give, in Theorem~\ref{quat_euclid_orthogonal}, the Euclidean dual code of any $Q_n$-code over $\Fq$ in the semisimple case. In doing so, we also correct a few mistakes appearing in \cite{CaoCaoFu23, GaoYue21} (see Remark~\ref{remark_mistake} and Example~\ref{ex-error}). Further, since the semisimple group algebras $\Fqq[Q_n]$ and $\Fqq[D_{2n}]$ are isomorphic (see \cite{Flaviana2009} or Theorem~\ref{thm:alg_isom} below), then the Hermitian dual of any $Q_n$-code can also be determined. Although Hermitian self-orthogonal $D_n$-codes and Euclidean self-orthogonal $Q_n$-codes were already determined in \cite{CaoCaoFu23} and \cite{CaoCaoFu16}, respectively, both papers are based on the concatenate structure of these quasi-cyclic codes, whilst our group-algebra approach substantially reduces the computational effort required (see Example~\ref{ex_cao}). Finally, based on the Hermitian CSS construction, we provide in section~\ref{sec:examples} numerical examples of some already known quantum error-correcting codes that attain the best known minimum distance for their length and dimension, and which may be systematically rebuilt from Hermitian self-orthogonal dihedral codes.


\section{Preliminaries}

All groups considered in this paper are supposed to be finite. We denote by $\Fq$ the finite field of $q$ elements, where $q$ is a power of a prime $p$, and ${\cal M}_n(R)$ for the ring of $(n\times n)$-matrices over a ring $R$. Moreover, $\langle X\rangle_R$  denotes the (left) ideal of $R$ generated by a non-empty subset $X$ of $R$ (we will simply write $\langle X\rangle$ when the ambient ring $R$ is clear enough). The rest of the notation not explained here is standard in the context of group theory or coding theory.

\subsection{Group codes}\label{subsec:groupcodes}

A \emph{linear code} $\cC\subseteq \Fq^n$ is just an $\Fq$-vector subspace of $\Fq^n$. We assume that the reader is familiar with the fundamentals of linear codes, which can be found in any basic book on coding theory. As mentioned in the introduction, under certain conditions, some linear codes can be viewed as ideals of a group algebra, thus taking advantage of this algebraic structure for their analysis. Let us define in a slightly different way, following \cite{bernal2009intrinsical}, this family of linear codes. Given a group $G=\{g_1,\dots ,  g_n\}$, we say that a linear code $\cC\subseteq \Fq^n$ is a (left) {\em $G$-code} if there exists a bijection $\phi: \{1,\dots, n\} \longrightarrow G$ such that the set
\begin{equation}\label{eq:Gcodi}
I_{\cC}:=\left\{\sum_{i=1}^n a_i \phi(i)\ |\  (a_1, \dots, a_n) \in  \cC \right\}
\end{equation}
is a (left) ideal of the group algebra $\Fq[G]$. In practice, we say that a linear code $\cC$ over $\Fq$ is a \emph{group code} if there exists a finite group $G$ such that $\cC$ is a $G$-code. As the reader may expect, the features of the group algebra $\Fq[G]$ become thus fundamental in the study of $G$-codes.

Recall that the group algebra $\Fq[G]$ is said to be {\em semisimple} if it can be realised as a direct sum of simple (also called irreducible) $\Fq[G]$-modules, and this occurs, in virtue of Maschke's theorem, if and only if the characteristic of the field does not divide the order of $G$ (\emph{cf.} \cite{doerk1992}). In this case, the well-known Wedderburn-Artin theorem (\cite[B - Theorem 4.4]{doerk1992}) ensures that $\Fq[G]$ can be realised as a direct sum of matrix algebras. 

\begin{theorem}\emph{(Wedderburn-Artin decomposition for finite group algebras).}
 \label{Wedderburn}
    Let $G$ be a finite group such that $\Fq[G]$ is a semisimple group algebra. Then $\Fq[G]$ is isomorphic, as $\Fq$-algebra, to the direct sum of some matrix rings over suitable extensions of $\Fq$. Specifically, one has: 
    \[
        \Fq[G] \cong \bigoplus_{i=1}^s {\cal M}_{n_i} \left( \Fqr{r_i} \right)
    \]
    satisfying $|G| = \displaystyle\sum_{i=1}^{s} n_i^2 r_i$.
\end{theorem}

Let $\mathcal{C}$ be a $G$-code, and $I_{\mathcal{C}}$ be the corresponding ideal in $\Fq[G]$ given in (\ref{eq:Gcodi}). Notice that when $\Fq[G]$ is semisimple, by the previous theorem, we can always write $I_{\mathcal{C}}$ as a (unique) sum $I_1\oplus \cdots \oplus I_s$, where $I_i$ is an ideal of $\mathcal{M}_{n_i}(\Fqr{r_i})$ for every $1\leq i \leq s$. Moreover, for every $1\leq i \leq s$, each $\mathcal{M}_{n_i}(\Fqr{r_i})$ is a principal ideal ring, so each ideal $I_i$ is generated by a unique matrix $M_i$ in row reduced echelon form, that is, $$I_i=\langle M_i\rangle = \{XM_i \; | \; X\in \mathcal{M}_{n_i}(\Fqr{r_i})\}.$$ In particular, the dimension of each ideal $I_i$ as $\Fq$-vector space is just $n_ir_i\operatorname{rank}(M_i)$, and consequently the sum of these dimensions yields the dimension of the group code $\mathcal{C}$ (see \cite[Lemma 3.7, Theorem 3.8]{Nostre24}), that is:
\begin{equation}\label{eq:dimension}
\operatorname{dim}_{\Fq}(\mathcal{C}) = \displaystyle\sum_{i=1}^s n_ir_i\operatorname{rank}(M_i). 
\end{equation}
This approach, which will be used without any further comment, is crucial along the paper, since it allows us to give all possible $G$-codes and their dimensions whenever the Wedderburn-Artin decomposition of $\Fq[G]$ is known. Nonetheless, the explicit computation of the decomposition of $\Fq[G]$ for some concrete group $G$ is a challenging problem. In the next subsection we review a few families of groups for which this decomposition has been computed.

\subsection{The Wedderburn-Artin decomposition of certain group algebras}

When the group algebra  $\F_{q}[G]$ is semisimple, the decomposition provided by the Wedderburn-Artin theorem has been concretised in certain cases. The simplest one is when $G=C_n$ is a cyclic group of order $n$. If $\car(\Fq)$ does not divide $n$, then we can factorise the polynomial $\mx^n-1$  as a product of irreducible polynomials over $\Fq[\mx]$, that is, $\mx^n-1=f_1f_2\dots f_r$, in order to apply the Chinese Remainder Theorem and obtain

\[
  \Fq[C_n] \cong \dfrac{\Fq[\mx]}{\langle\mx^n-1\rangle} \cong  \bigoplus_{i=1}^{r} \dfrac{\Fq[\mx]}{\langle f_i\rangle} \cong \bigoplus_{i=1}^{r} \Fqr{\deg(f_i)}.
\]

In the more general case of having an abelian group, the Wedderburn-Artin decomposition is also known (see \cite{Perlis1950}). Nevertheless, in this article we deal with non-abelian groups; specifically, we are interested in decompositions associated with dihedral groups and generalised quaternion groups, and for that reason we detail them below. 

With regard to dihedral groups $D_n$ of $2n$ elements, such a decomposition is given in \cite{Bro15}. For every non-zero polynomial $g\in \Fq[\mx]$, we denote by $g^*$ the {\em reciprocal} polynomial of $g$, i.e. $g^*(\mx) :=g(0)^{-1}\mx^{deg(g)}g(\mx^{-1})$.  The polynomial $g$ is said to be {\em self-reciprocal} if $g=g^*$. With this notation, the polynomial $\mx^n-1$ can be factorised over $\Fq[\mx]$ into irreducible monic polynomials as 
\begin{equation}
\label{decomp_polinomi_inicial}
\mx^n-1 = f_1f_2 \cdots f_r f_{r+1}f_{r+1}^* \cdots f_{r+s}f_{r+s}^*,
\end{equation}
where $f_1 := \mx-1$, $f_2:=\mx+1$ if $n$ is even, and $f_i=f_{i}^*$ for $1\leq i\leq r$. In this way, $r$ is the number of self-reciprocal factors in the factorisation and $2s$ the number of factors that are not self-reciprocal. Set $\zeta(n) := 2$ if $n$ is even and $\zeta(n) := 1$ otherwise. Let $\alpha_i$ be any root of $f_i$. If $\car(\Fq)$ does not divide $2n$, then \cite[Theorem 3.1]{Bro15} asserts that there exists an algebra isomorphism
    \begin{equation}\label{eq:isom_decom_diedric}
    \Fq[D_n] \overset{\rho}{\cong} \bigoplus_{i=1}^{r+s}A_i,
    \end{equation}
    where
    \begin{equation}\label{eq:decom_diedric}
    A_i:= 
    \begin{cases}
        \Fq \oplus \Fq & \text{if } 1 \leq i \leq \zeta(n) \\[1mm] 
        {\cal M}_2 \left(\Fq(\alpha_i+\alpha_i^{-1}) \right) & \text{if } \zeta(n) +1\leq i \leq r \\[2mm]
         {\cal M}_2 \left(\Fq(\alpha_i) \right) & \text{if } r+1 \leq i \leq r + s
    \end{cases}.
    \end{equation}

\begin{remark}\label{arrels_recipr}
Clearly, if $\alpha$ is a root of a polynomial $f\in \Fq[\mx]$,  then $\alpha^{-1}$ is a root of $f^*$ since $f^{*}(\alpha^{-1})=f(0)^{-1}\alpha^{-\operatorname{deg}(f)}f(\alpha)=0$. It follows that when $f$ is self-reciprocal and $\pm 1$ are not roots of $f$, then we can write $f$ as product of $(\mx-\alpha)(\mx-\alpha^{-1})$ for some roots $\alpha$ of $f$, and therefore $f$ has even degree.  As a consequence, polynomials $f_i$ of $(\ref{decomp_polinomi_inicial})$ with $\zeta(n)+1\leq i \leq r$ always have even degree, and moreover $\Fq(\alpha_i+\alpha_i^{-1})\cong \Fqr{\deg(f_i)/2}$ (see \cite[Remark 3.2]{Bro15}).
\end{remark}

Beyond the semisimple case, in the same paper, we find a similar decomposition of $\Fq[D_n]$ when $\car(\Fq)=2\nmid n$. In this case, $\zeta(n)=1$ and the isomorphism provided in (\ref{eq:isom_decom_diedric}) works whenever one simply replaces $A_1$ by $\Fq[C_2]$ (see \cite[Remark 3.4]{Bro15}).  In this way, for dihedral groups, we will not restrict ourselves just to the semisimple case, but we will also consider this situation, thus allowing, in particular, the construction of binary codes. We will make extensive use of the isomorphism $\rho$ provided in $(\ref{eq:isom_decom_diedric})$ in both sections 3 and 4, where we will explain its definition in detail (see Theorems \ref{theo_dih_refined} and \ref{theo_quat_GaoYue}).

Based on the ideas within \cite{Bro15}, the authors of \cite{GaoYue21} computed the decomposition of the semisimple group algebra $\Fq[Q_n]$ for generalised quaternion groups $Q_n$; for convenience, we will thoroughly describe that Wedderburn-Artin decomposition in section~\ref{sec:quaternion}. There are other families of groups for which the decomposition of the corresponding group algebra over finite fields is also known. For instance, by adapting known results about $\mathbb{Q}[S_n]$ and $\mathbb{Q}[A_n]$, the Wedderburn-Artin decompositions of $\Fq[S_n]$ and $\Fq[A_n]$ are obtained in \cite{Ricardo2023}. In \cite{Gao2020} and \cite{Brochero2022} it is computed the Wedderburn-Artin decomposition for generalised dihedral group algebras and some metacyclic group algebras, respectively. Furthermore, in \cite{VedDeu19} it is also considered the group algebra of the direct product of two dihedral groups $D_n \times D_m$ of orders $2n$ and $2m$, respectively, such that $m$ divides $q-1$. This last result has been recently generalised by ourselves in \cite{Nostre24}, where we give a closed formula for the Wedderburn-Artin decomposition of the group algebra $\Fq[G\times H]$ corresponding to the direct product of two groups $G$ and $H$ based on the decompositions of $\Fq[G]$ and $\Fq[H]$; as a consequence, the hypothesis ``$m$ divides $q-1$'' in the aforementioned result concerning the structure of $\Fq[D_n\times D_m]$ in \cite{VedDeu19} can be removed.


\subsection{Dualities of group codes}
\label{subsec:dualgroupcode}

Given two elements  $x=(x_1,x_2,...,x_n)$ and $y=(y_1,y_2,...,y_n)$ in $\mathbb{F}_q^n$, their {\em Euclidean inner product} is defined as $(x|y)_E:=\sum_{i=1}^n x_i y_i$.  Thus, the {\em Euclidean dual code} of a linear code $\mathcal{C}$ of $\Fq^n$ is 
\begin{equation}\label{eq:dualEuc}
\mathcal{C}^{\perp_E}:=\{y\in\mathbb{F}_q^n \; \mid \; (x|y)_E =0, \;\forall \,  x\in\mathcal{C}\}.
\end{equation}
Note that this inner product may be transferred to the group algebra $\Fq[G]$, which is useful to work with group codes. More concretely, if $a,b\in\Fq[G]$, then
\[
(a|b)_E=\left(\sum_{g\in G}a_g g \:\middle\vert\: \sum_{g\in G}b_g g\right)_E:=\sum_{g\in G}a_g b_g
\]
is defined as the \emph{Euclidean inner product}, and the \emph{Euclidean dual} of a left ideal $I\subseteq \Fq[G]$ is $I^{\perp_E}:=\{b\in\Fq[G]\; | \; (a|b)_E=0, \: \forall\, a\in I\}$. 

We claim that $I^{\perp_E}$ is also a left ideal of $\Fq[G]$. Certainly, since the Euclidean inner product of $\Fq[G]$ is a symmetric bilinear form, it is enough to prove that $gb\in I^{\perp_E}$ for any $b\in I^{\perp_E}$ and for any $g\in G$; but this follows from the fact that multiplying $a,gb\in \Fq[G]$ by $g^{-1}$ simply produces in $(a|gb)_E$ a permutation of its summands, so in particular $(a|gb)_E=(g^{-1}a|b)_E=0$ for all $a\in I$ because $g^{-1}a\in I$.

Let $\cC$ be a $G$-code, and let $I_{\cC}$ be the corresponding ideal of $\Fq[G]$ (see (\ref{eq:Gcodi})). With the previous inner product in $\Fq[G]$, it is straightforward to see that the ideal $I_{\cC^{\perp_E}}$ of $\Fq[G]$ associated with $\cC^{\perp_E}$ is effectively $(I_{\cC})^{\perp_E}$, that is, $I_{\cC^{\perp_E}}=(I_{\cC})^{\perp_E}$. In particular, we have deduced that the Euclidean dual code of a $G$-code is also a $G$-code.

Another way to compute $I^{\perp_E}$ based on the multiplication in $\Fq[G]$, which will be crucial for our purposes, is given in \cite{Borello22}. To present it, we need to consider the algebra antiautomorphism of $\Fq[G]$, given by $\hat{a}:=\sum_{g\in G}a_g g^{-1}$ for all $a=\sum_{g\in G}a_g g$, and the concept of right annihilator of a left ideal $I$ in $\Fq[G]$:
\[
\operatorname{Ann}_r(I):=\{b\in \Fq[G]\ |\ ab=0, \forall a\in  I\},
\]
which is a right ideal of $\Fq[G]$. Now by \cite[Lemma 2.5]{Borello22} we get that
\begin{equation}
\label{eq:orthogonal_ideal_Euclidean}
I^{\perp_E}=\widehat{\operatorname{Ann}_r(I)}:=\{ \hat{b}\in \Fq[G]\ |\ ab=0, \forall a\in  I\}.
\end{equation}

In this paper we are going to consider also the Hermitian metric. Given $x=(x_1,x_2,...,x_n)$ and $y=(y_1,y_2,...,y_n)$ in $\mathbb{F}_{q^2}^n$, their {\em Hermitian inner product} is defined as $(x|y)_H:=\sum_{i=1}^n x_i y_i^q$.  Thus, the {\em Hermitian dual code} of a linear code $\mathcal{C}\subseteq\Fqq^n$ is 
\[
\mathcal{C}^{\perp_H}:=\{y\in\mathbb{F}_{q^2}^n \; \mid \; (x|y)_H =0, \;\forall \,  x\in\mathcal{C}\}.
\]
We can also transfer the Hermitian metric to the group algebra. Analogously, given $a, b\in \Fqq [G]$, their \emph{Hermitian inner product} is
\[
(a|b)_H=\left(\sum_{g\in G}a_g g \:\middle\vert\: \sum_{g\in G}b_g g\right)_H:=\sum_{g\in G}a_g b_g^q,
\]
and the \emph{Hermitian dual} of a left ideal $I\subseteq \Fqq[G]$ is $I^{\perp_H}:=\{b\in\Fqq[G]\; | \; (a|b)_H=0, \: \forall\, a\in I\}$, which turns out to be also a left ideal of $\Fqq[G]$. As in the Euclidean case, here we also have that given a linear code $\mathcal{C}\subseteq\Fqq^n$, the ideal $I_{\cC^{\perp_H}}\subseteq\Fqq [G] $ associated with $\cC^{\perp_H}$ is equal to $(I_{\cC})^{\perp_H}$. The code $\cC$ and the corresponding ideal $I_{\cC}$ are called {\em Hermitian self-orthogonal} whenever $\cC\subseteq \mathcal{C}^{\perp_H}$, or equivalently whenever $I_{\cC}\subseteq I_{\cC^{\perp_H}}=(I_{\cC})^{\perp_H}$.

Finally, consider the automorphism $x\mapsto x^q$ of $\Fqr{2}$. It is not difficult to see from (\ref{eq:dualEuc}) that $\mathcal{C}^{\perp_H}$ can be computed from $\mathcal{C}^{\perp_E}$ by simply applying that automorphism to every component of each codeword of $\mathcal{C}^{\perp_E}$, that is, 
$$\mathcal{C}^{\perp_H}=(\mathcal{C}^{\perp_E})^q:=\{(y_1^q,...,y_n^q) \; | \; (y_1,...,y_n)\in C^{\perp_E}\}.$$
Similarly, to work in the group algebra, we denote by $\operatorname{conj}_q$ the natural linear extension to $\Fqr{2}[G]$ of the automorphism $x\mapsto x^q$ of $\Fqr{2}$, that is, 
\begin{equation}\label{defConj}
 \operatorname{conj}_q\left(\sum_{g\in G}b_g g\right):=\sum_{g\in G}b_g^q g\in \Fqr{2}[G]. 
\end{equation}
Therefore, one has 
\begin{equation} \label{eq:orthogonal_ideal_Hermitian}
(I_{\cC})^{\perp_H}=((I_{\cC})^{\perp_E})^q=(\widehat{\operatorname{Ann}_r(I_{\cC})})^q:=\{ \operatorname{conj}_q(\hat{b})\in \Fqq[G]\ |\ ab=0, \forall a\in  I_{\cC}\}.
\end{equation}

Finally, notice that when the group algebra can be decomposed as a direct sum of matrix algebras over finite fields (see Theorem \ref{Wedderburn}), obtaining the dual of $I_{\cC}$, either Euclidean or Hermitian, can then be reduced to find the dual of each summand of $I_{\cC}$ in the corresponding matrix ring that appears in the decomposition of the group algebra.  In other words, 
if $I_{\cC}\cong I_{1}\oplus \dots\oplus I_{s}$, where $I_{i}$ is an ideal of ${\cal M}_{n_{i}}(\Fqr{r_{i}})$, for every $1\leq i\leq s$, then  
\begin{equation}\label{reduccioDual}
I_{{\cC}^{\perp_E}}=(I_{\cC})^{\perp_E}\cong I_{1}^{\perp_E}\oplus \dots\oplus I_{s}^{\perp_E} \quad \mbox{ and } \quad 
I_{{\cC}^{\perp_H}}=(I_{\cC})^{\perp_H}\cong I_{1}^{\perp_H}\oplus \dots\oplus I_{s}^{\perp_H}.
\end{equation}
This will be the strategy used hereafter to construct dualities.


\subsection{On the Euclidean dual of a \texorpdfstring{$D_n$-code}{Dn-code}}
\label{Euclidean_dihedral}

As already mentioned in the introduction, Vedenev and Deundyak utilised in \cite{VedDeu21} the Wedderburn-Artin decomposition of $\Fq[D_n]$ through the isomorphism $\rho$ given in (\ref{eq:isom_decom_diedric}) to provide a precise description of the Euclidean dual of a $D_n$-code. To do this, they break down the ideal associated with the code into ideals in each matrix ring $A_i$ (see (\ref{eq:decom_diedric})) and then apply the equality described in (\ref{reduccioDual}). In such a way, the situation is reduced to the computation of the Euclidian dual of each ideal $ I_i$ in the corresponding matrix ring $A_i$. We collect this information in the theorem below, and we adapt it to the notation used in this paper. We recall that, by \cite[Remark 3.4]{Bro15}, in case $\car(\Fq)=2 \nmid n$ we set $A_1=\Fq[C_2]$ in (\ref{eq:isom_decom_diedric}), and we take $C_2$ as the cyclic group of order $2$ generated by the matrix 
$\left[\begin{smallmatrix} 
0&1\\1&0\end{smallmatrix}\right]$. It is easy to check that the only proper ideal of $A_1$ in this case is the one generated by $\left[\begin{smallmatrix} 
   1&1\\1&1\end{smallmatrix}\right]$.

\begin{theorem}\emph{(\cite[Theorem 5]{VedDeu21} rephrased).}
\label{dih_Euclidean_dual}
Assume that $\car(\Fq) \nmid n$. Let $\mathcal{C}$ be a $D_n$-code over $\Fq$, and $I_{\cC}$ be its corresponding ideal in $\mathbb{F}_q[D_n]$. Using the isomorphism $\rho$ given in \emph{(\ref{eq:isom_decom_diedric})} and \emph{(\ref{eq:decom_diedric})}, the ideal $I_{\cC}$ can be expressed as 
$$
I_{\cC} \  \overset{\rho}{\cong} \  \displaystyle \bigoplus_{i=1}^{r+s} 
I_i  \subseteq \displaystyle \bigoplus_{i=1}^{r+s} A_i.
$$
Set $I_i= \langle M_i\rangle $, 
for $\zeta(n)+1\leq i \leq r$, where $M_i$ is a matrix of $A_i$ in row reduced echelon form. Then $I_{\cC^{\perp_E}}\subseteq \mathbb{F}_q[D_n]$ verifies that 
$$
I_{\cC^{\perp_E}}  \  \overset{\rho}{ \cong} \  \displaystyle \bigoplus_{i=1}^{r+s}
I_i^{\perp_E} \subseteq \displaystyle \bigoplus_{i=1}^{r+s}A_i
$$ 
where: 
\begin{enumerate}
\item[\emph{i)}] for $1\leq i\leq \zeta(n)$ and $\car(\Fq)\neq 2$: 
$$
I_i^{\perp_E} =\left\lbrace\begin{array}{ll}
\Fq\oplus \Fq, & \text{if } 
I_i=\pmb{0}\oplus \pmb{0} \\*[4mm]
\Fq\oplus \pmb{0}, & \text{if } 
I_i=\pmb{0}\oplus \Fq \\*[4mm]
\pmb{0}\oplus \Fq, & \text{if } 
I_i=\Fq\oplus \pmb{0} \\*[4mm]
\pmb{0}\oplus \pmb{0}, & \text{if } 
I_i=\Fq\oplus \Fq \\*[4mm]
\end{array}\right. 
$$

\item[\emph{ii)}] for $\car(\Fq)= 2$ it holds: 

$$
I_1^{\perp_E} =\left\lbrace\begin{array}{ll}
\Fq[C_{2}], & \text{if } I_1=\pmb{0} \\*[4mm]
\left\langle \begin{bmatrix} 1&1\\1&1 \end{bmatrix}\right\rangle_{\Fq[C_{2}]}, & \text{if } 
I_1=\left\langle \begin{bmatrix} 1&1\\1&1 \end{bmatrix}\right\rangle_{\Fq[C_{2}]} \\*[4mm]
\pmb{0}, & \text{if } I_1=\Fq[C_{2}] \\*[4mm]
\end{array}\right. 
$$

\item[\emph{iii)}] for $\zeta(n)+1\leq i \leq r$, a root $\alpha_i$ of $f_i$, and $\lambda_i\in \Fq(\alpha_i+\alpha_i^{-1})$ it holds $I_i^{\perp_E}=\langle M_i\rangle^{\perp_E}=\langle \widetilde{M_i}\rangle$ with one of the following options:
$$\widetilde{M}_i:=\left\lbrace\begin{array}{ll}
\begin{bmatrix} 1&0\\0&1 \end{bmatrix}, & \text{if } M_i=\begin{bmatrix} 0&0\\0&0 \end{bmatrix} \\*[4mm]
\begin{bmatrix} 2&-\alpha_i-\alpha_i^{-1}\\0&0 \end{bmatrix}, & \text{if } M_i=\begin{bmatrix} 0&1\\0&0 \end{bmatrix} \\*[4mm]
\begin{bmatrix} \alpha_i+\alpha_i^{-1}+2\lambda_i &-2-(\alpha_i+\alpha_i^{-1})\lambda_i\\0&0 \end{bmatrix}, & \text{if } M_i=\begin{bmatrix}1&\lambda_i\\0&0 \end{bmatrix} \\*[4mm] 
\begin{bmatrix} 0&0\\0&0 \end{bmatrix}, & \text{if } M_i=\begin{bmatrix} 1&0\\0&1 \end{bmatrix}
\end{array}\right. $$

\item[\emph{iv)}] for $r+1\leq i \leq r+s$, and $\lambda_i\in \Fq(\alpha_i)$ it holds $I_i^{\perp_E}=\langle M_i\rangle^{\perp_E}=\langle \widetilde{M_i}\rangle$ with one of the following options:
$$\widetilde{M}_i:=\left\lbrace\begin{array}{ll}
\begin{bmatrix} 1&0\\0&1 \end{bmatrix}, & \text{if } M_i=\begin{bmatrix} 0&0\\0&0 \end{bmatrix} \\*[4mm]
\begin{bmatrix} 0&1\\0&0 \end{bmatrix}, & \text{if } M_i=\begin{bmatrix} 0&1\\0&0 \end{bmatrix} \\*[4mm] 
\begin{bmatrix}1&-\lambda_i\\0&0 \end{bmatrix}, & \text{if } M_i=\begin{bmatrix}1&\lambda_i\\0&0 \end{bmatrix}\\*[4mm] 
\begin{bmatrix} 0&0\\0&0 \end{bmatrix}, & \text{if } M_i=\begin{bmatrix} 1&0\\0&1 \end{bmatrix}
\end{array}\right. $$
\end{enumerate}
\end{theorem}


\subsection{CSS quantum codes}\label{pre:CSS}

A $[[n, k, d]]_q$ {\em quantum error-correcting code} ${\cal Q}$ is a subspace of the $n$-fold tensor product $(\bbC^q)^{\otimes n}$ of the complex vector space $\bbC^q$. The code ${\cal Q}$ has length $n$, dimension $\dim_{\bbC^q}({\cal Q}) = q^k$,  and minimum distance $d$, \emph{i.e.} any error acting on at most $d -1$ positions of the tensor factors (the so-called \emph{qubits}) can be detected or has no effect on the code. This kind of codes are useful to reduce the effects of environmental and operational noise in quantum information processing. Here we focus on the so-called \emph{stabiliser codes} based on linear codes over finite fields and, within them, we concretely deal with {\em quantum CSS codes}, whose name is due to Calderbank, Shor and Steane. As already mentioned in the Introduction, quantum CSS codes can be constructed from classical linear codes in two ways: either using the Euclidean metric and two linear nested codes (or even an Euclidean self-orthogonal code), or using an Hermitian self-orthogonal code. The basic theory of this kind of codes can be found on \cite{KKKK}, and for more information on general quantum codes, see, for example, \cite{Laguardia20}. 

In \cite{Nostre24} we made use of the algebraic representation of the Euclidean dual of a $D_n$-code over $\Fq$ given in \cite{VedDeu21}, which is based on the Wedderburn-Artin decomposition of $\Fq[D_n]$, to build quantum CSS dihedral codes. However, we were not able to obtain optimal quantum codes; but on the upside, our systematic approach avoided brute force computations to check whether the considered pair of $D_n$-codes were nested. Note that the Hermitian CSS construction of quantum codes involves computations in larger fields, which often yield better parameters in experimental results. Thus, in this paper, one of our main objectives is to apply that construction to Hermitian self-orthogonal $D_n$-codes over $\Fqq$ that can be produced methodically via Theorem \ref{dih-self-hermitic}. 

Let us present below the CSS construction of quantum codes based on the Hermitian metric. In the following theorem, $\operatorname{wgt}(c)$ denotes the Hamming weight of a codeword $c$ in a linear code $\cC$.

\begin{theorem}\emph{(\cite[Proposition II.4]{Grassl24})}\label{theo:CSS}
Let $\cC\subseteq \Fqq^n$ be an Hermitian self-orthogonal linear code, that is, $\cC\subseteq \mathcal{C}^{\perp_H}$, and suppose that $\cC$ has dimension $k$. Then there exists an  $[[n, n-2k, d_{\cal Q}]]_q$ quantum stabiliser code ${\cal Q}$ with \[
d_{\cal Q}:=\min\{\operatorname{wgt}(x-y) \; | \; x,y\in {\cal C}^{\perp_H}\smallsetminus \cC,\:x\neq y\} \geq d
\] if $\cC \neq \mathcal{C}^{\perp_H}$ and $d$ is the minimum distance of $\mathcal{C}^{\perp_H}$, and otherwise $d_{\cal Q}:=d$.
\end{theorem}


\section{On Hermitian dualities of dihedral codes}

As aforementioned, Cao \emph{et al.} obtained expressions of $D_n$-codes, their Euclidean dual codes (over $\Fq$) and their Hermitian dual codes (over $\Fqq$) via the concatenated structure of dihedral codes (see \cite{CaoCaoMa22, CaoCaoFu23}) when $q$ does not divide $n$. Alternatively, Vedenev and Deundyak also analysed in \cite{VedDeu21} the structure of any $D_n$-code and its Euclidean dual over $\Fq$, with $\car(\Fq)\nmid n$, but with a different approach: through the structure of the corresponding ideal in the group algebra $\Fq[D_n]$. Following Vedenev and Deundyak's spirit, our objective in this section is to obtain a complete algebraic description of the Hermitian dual code of any $D_n$-code $\cC$ based on the decomposition of the corresponding ideal $I_{\cC}$ in the group algebra $\Fqq[D_n]$, whenever $\car(\Fqq)\nmid n$. The main result of this section is Theorem~\ref{dih_hermitic_orthogonal}. Additionally, we have detected some incorrect statements in \cite{CaoCaoFu23} (see Remark~\ref{remark_mistake}).

\subsection{A suitable Wedderburn-Artin decomposition of \texorpdfstring{$\F_{q^2}[D_n]$}{Fq2[Dn]}}

In this section we are going to rewrite appropriately the isomorphism given in (\ref{eq:isom_decom_diedric}) to address the study of the Hermitian dual code of any dihedral code over $\Fqq$. To do this, we need to refine the decomposition given in (\ref{decomp_polinomi_inicial}) of the polynomial $\mx^n-1$ as a product of irreducible polynomials over $\Fqq[\mx]$. In this case, in addition to the reciprocal polynomial, we will also take into account the conjugate polynomial of a given one. 

Given a polynomial $f=\sum_{i=0}^m a_i\mx^i\in\Fqr{2}[\mx]$, we denote by $\overline{f}:=\sum_{i=0}^m a_i^q\mx^i$ its \emph{conjugate polynomial}. In particular, $f$ is called \emph{self-conjugate} whenever $f=\overline{f}$. We write $f^{\dagger}$ for the conjugate of the reciprocal polynomial of $f$, that is, $f^{\dagger}:=\overline{f^{*}}$. 

Below we state some elementary properties about polynomials $f$, $f^*$ and $\overline{f}$ which will be useful later.

\begin{lemma}
\label{propiedades_reci_conj}
Let $\alpha$ be an element in some field extension of $\Fqq$. Then for any polynomial  $f\in\Fqr{2}[\mx]$ it holds: 
\begin{itemize}
\setlength{\itemsep}{-1mm}
\item[\emph{(a)}] $\overline{f}(\alpha^q)=f(\alpha)^q$. In particular, $\overline{f}(1)=f(1)^q$, $\overline{f}(-1)=f(-1)^q$, and if $\alpha$ is a root of $f$, then $\alpha^q$ is a root  of $\overline{f}$.
\end{itemize}
Assume, in addition, that $f$ is irreducible on $\Fqr{2}[\mx]$.
\begin{itemize}
\setlength{\itemsep}{-1mm}
\item[\emph{(b)}] If  $f=\overline{f}$,  then $f\in\Fq[\mx]$, and $\deg(f)$ is odd.
\item[\emph{(c)}] If  $f^{*}=\overline{f}$, then $\deg(f)$ is odd; and if $\alpha$ is a root of $f$, then $\alpha^{-1}=\alpha^{q^{\deg(f)}}$.
\end{itemize}
\end{lemma}

\begin{proof}
Let us write $f=\displaystyle\sum_{i=0}^r a_i\mx^i\in \Fqq[\mx]$.

\noindent (a) $\overline{f}(\alpha^q)=\displaystyle\sum_{i=0}^r a_i^q(\alpha^q)^i=\sum_{i=0}^r (a_i\alpha^i)^q = f(\alpha)^q$, and the remaining assertions directly follow.

\medskip

\noindent (b) The first assertion follows immediately from the fact that $f$ is self-conjugate, i.e. $a_i=a_i^q$ for every $0\leq i\leq r$. Set $r:=\deg(f)$, and let $\alpha$ be a root of $f$. We claim that $r$ is odd. Since $f$ irreducible over $\Fqq[\mx]$, one has that $\Fqr{2}(\alpha)\cong \Fqr{2r}$ is the splitting field of $f$ over $\Fqq$. Nevertheless, since $f\in \Fq[\mx]$ it  is also irreducible over $\Fq[\mx]$, then $\Fq(\alpha)\cong\Fqr{r}$ is the splitting field of $f$ over $\Fq$. By contradiction, if $r$ is even, we have $\Fq \leqslant \Fqr{2} \leqslant \Fq(\alpha)=\Fqr{2}(\alpha)$ and  so
$
r=[\Fq(\alpha):\Fq]=[\Fqr{2}(\alpha):\Fq]=[\Fqr{2}(\alpha):\Fqr{2}][\Fqr{2}: \Fq]=2r
$, 
a contradiction.

\medskip

\noindent (c) In virtue of statement (a) and Remark \ref{arrels_recipr}, we get that $\alpha^{-1}$ and $\alpha^q$ are both roots of $f^{*}=\overline{f}$. Hence there exists a $k$-th power of the automorphism $\sigma: x\mapsto x^{q^2}$ such that $\alpha^{-1}=\sigma^k(\alpha^q)=\alpha^{q^{2k+1}}$. Set $m:=2k+1$. Note that we may assume $k<r=\deg(f)$, since otherwise $\alpha^{-1}=\alpha^{q^{2r+1}}=\alpha^{q}$ and thus $\alpha^{q^2}=\alpha\in\Fqr{2}$, so $r=1$ and the claim certainly holds. Now observe that $\alpha^{q^m}=\alpha^{-1}$ leads to $\alpha^{q^{2m}}=\alpha$, so $\alpha \in \Fqr{2m}$. Since all roots of $f$ are powers of $\alpha$, then $f$ splits over $\Fqr{2m}$, and so $\Fqr{2r}\leqslant \Fqr{2m}$ because $\Fqr{2r}$ is the splitting field of $f$ over $\Fqr{2}$. In particular $r$ divides $m$. The fact $m<2r$ leads to $m=r$, which finishes the proof.
\end{proof}

\medskip

\begin{hypotheses} 
\label{hipotesis_descomposicion_Fq2Dn}
Suppose that $\car(\Fqq)$ does not divide $n$,  and decompose the polynomial $\mx^n-1$ into irreducible factors in $\Fqr{2}[\mx]$ as follows:
\begin{equation}
\label{decomp_polinomi_detall}
\mx^n-1 = \prod_{j \in J_0} f_j
\prod_{j \in J_1} f_j \overline{f_j}
\prod_{j \in J_2 \cup J_3} f_j f_j^*
\prod_{j \in J_4} f_j f_j^* \overline{f_j} f_j^\dagger,
\end{equation}
where:
\begin{itemize}
\setlength{\itemsep}{-.5mm}
\item $j \in J_0$ if and only if $\overline{f_j} =f_j =  f_j^*$. In this case, $f_j \in \{ \mx-1, \mx+1 \}$ by \emph{Lemma~\ref{propiedades_reci_conj} (b)} and \emph{Remark~\ref{arrels_recipr}}.
\item $j \in J_1$ if and only if $\overline{f_j}\neq f_j = f_j^*$;
\item $j \in J_2$ if and only if  $\overline{f_j}= f_j \neq  f_j^*$;
\item $j \in J_3$ if and only if  $\overline{f_j} = f_j^*\neq f_j$;
\item $j \in J_4$ if and only if $f_j, \overline{f_j}$ and $f_j^*$ are all different.
\end{itemize}
Moreover, set $J := \displaystyle\bigcup_{i=0}^4 J_i$, $\alpha_j$ for a root of $f_j$, and $r_j:=\deg(f_j)$.
\end{hypotheses}

\medskip

In the following lemma, we explain the relationship between an arbitrary polynomial of $\Fqr{2}[\mx]$ and the roots of the polynomials $f_j$ we have just defined.

\begin{lemma}
\label{lemma_tech_rj}
Assume \emph{Hypotheses~\ref{hipotesis_descomposicion_Fq2Dn}}. For any $g\in\Fqr{2}[\mx]$ it holds:
\begin{itemize}
\setlength{\itemsep}{-1mm}
\item[\emph{(a)}] If $j\in J_1\cup J_4$, then $\overline{g}(\alpha_j)=g(\alpha_j^q)^{q^{2r_j-1}}$.
\item[\emph{(b)}] If $j\in J_2$, then $\overline{g}(\alpha_j)=g(\alpha_j)^{q^{r_j}}$.
\item[\emph{(c)}] If $j\in J_3$, then $\overline{g}(\alpha_j)=g(\alpha_j^{-1})^{q^{r_j}}$.
\end{itemize}
\end{lemma}

\begin{proof}
Let $g=\sum_{i=0}^s a_i\mx^i$. Before proving the statements, observe that for any positive integer $k$, and any value $\beta$ in some field extension of $\Fqq$, it holds 
\begin{equation} \label{eq_lemma_tech}
g(\beta^{q^{2k}})=\displaystyle\sum_{i=0}^s a_i \left( \beta^{q^{2k}} \right)^i = \displaystyle\sum_{i=0}^s (a_i)^{q^{2k}} \left( \beta^{i} \right)^{q^{2k}} = \left(\displaystyle\sum_{i=0}^s a_i \beta^{i}  \right)^{q^{2k}} = g(\beta)^{q^{2k}}
\end{equation} 
since the coefficients of $g$ are fixed by the automorphism $\mx \mapsto \mx^{q^2}$ of $\Fqq$. Now for showing (a), observe that $$\overline{g}(\alpha_j) = \overline{g}\left(\alpha_j^{q^{2r_j}}\right) = \overline{g}\left(\left(\alpha_j^{q^2}\right)^{q^{2(r_j-1)}}\right)  = \overline{g}\left(\alpha_j^{q^2}\right)^{q^{2(r_j-1)}} = g(\alpha_j^{q})^{q^{2r_j-1}}, $$ where the first equality is due to $\alpha_j \in \Fqr{2}(\alpha_j)\cong\Fqr{2r_j}$ (see Remark~\ref{remark_fields}), the third equality follows from (\ref{eq_lemma_tech}), and the last equality is a direct application of Lemma~\ref{propiedades_reci_conj} (a). 

Let us prove (b). By Lemma~\ref{propiedades_reci_conj} (b), we have that $\alpha_j=\alpha_j^{q^{r_j}}$  and $r_j=2k_j+1$ for certain positive integer $k_j$. Thus $$ \overline{g}(\alpha_j) = \overline{g}\left(\alpha_j^{q^{r_j}}\right)= \overline{g}\left(\alpha_j^{q}\right)^{q^{2k_j}} = g\left(\alpha_j\right)^{q^{2k_j+1}}=g(\alpha_j)^{q^{r_j}},$$ where the second equality follows by applying (\ref{eq_lemma_tech}), and the third equality is due to Lemma~\ref{propiedades_reci_conj} (a). 

Finally, note that a similar reasonament yields (c), since by Lemma~\ref{propiedades_reci_conj} (c) we also have that $r_j=2k_j+1$ for certain positive integer $k_j$, and in this case $\alpha_j^{-1}=\alpha_j^{q^{r_j}}$.
\end{proof}

Before rewriting the isomorphism given in (\ref{eq:isom_decom_diedric}) properly, we also need the following result.

\begin{lemma}
\label{lemma:sigmas}
Assume \emph{Hypotheses~\ref{hipotesis_descomposicion_Fq2Dn}}. If $j\in J_1$, then:
\begin{enumerate}[label=\emph{(\alph*)}]
\setlength{\itemsep}{-1mm}
\item $Z_j := \begin{bmatrix}
1 & -\alpha_j \\ 1 & -\alpha_j^{-1}
\end{bmatrix}$ and 
$\overline{Z_j} := \begin{bmatrix}
1 & -\alpha_j^q \\[2mm] 1 & -\alpha_j^{-q}
\end{bmatrix}$ are invertible matrices of ${\cal M}_2(\Fqq(\alpha_j))$.

\item The maps $\sigma_j, \overline{\sigma_j} : {\cal M}_2(\Fqq(\alpha_j))\to {\cal M}_2(\Fqq(\alpha_j))$ given by 
\[
\sigma_j(X) := Z_j^{-1} X Z_j , \quad \overline{\sigma_j}(X) := \overline{Z_j}^{-1} X \overline{Z_j}
\]
are automorphisms of ${\cal M}_2(\Fqq(\alpha_j))$.
\end{enumerate}
\end{lemma}
\begin{proof}
Certainly $Z_j$ has non-zero determinant; otherwise $\alpha_j=\alpha_j^{-1}$, so $\alpha_j^2=1$ and consequently $f$ is a divisor of $\mx^2-1$, which is impossible since $j\in J_1$. Similarly, $\overline{Z_j}$ is also invertible. Finally, the second statement follows from the first one due to the fact that $\sigma_j$ and $\overline{\sigma_j}$ are conjugation maps. 
\end{proof}

\medskip

Given the group algebra $\Fqq[D_n]$, we restate below the isomorphism $\rho$ given in (\ref{eq:isom_decom_diedric}), by considering the decomposition in (\ref{decomp_polinomi_detall}) and by grouping together those blocks associated to pairs of conjugate polynomials. In the next result, $C_2$ denotes the cyclic group of order $2$ generated by the matrix $\left[\begin{smallmatrix} 0&1\\1&0 \end{smallmatrix}\right]$.

\begin{theorem}
\label{theo_dih_refined}
Assume \emph{Hypotheses~\ref{hipotesis_descomposicion_Fq2Dn}}. There exists an isomorphism of $\Fqq$-algebras $\rho:\Fqr{2}[D_n] \longrightarrow \displaystyle\bigoplus_{j \in J} A_j$, where
$$ A_j := 
\begin{cases}
\Fqr{2} \oplus \Fqr{2} & \text{if } j \in J_0 \text{ and $\car(\Fqq)\neq 2$}  \\

\Fqr{2}[C_2] & \text{if } j \in J_0 \text{ and $\car(\Fqq)=2$}  \\

\mathcal{M}_2 \left( \Fqr{2}(\alpha_j + \alpha_j^{-1}) \right) \oplus \mathcal{M}_2 \left( \Fqr{2}(\alpha_j^q + \alpha_j^{-q}) \right) & \text{if } j \in J_1 \\

\mathcal{M}_2 \left( \Fqr{2}(\alpha_j) \right) & \text{if } j \in J_2 \cup J_3 \\

\mathcal{M}_2 \left( \Fqr{2}(\alpha_j) \right) \oplus \mathcal{M}_2 \left( \Fqr{2}(\alpha_j^q)\right) & \text{if } j \in J_4 \\
\end{cases}.$$ Further, $\rho=\bigoplus_{j \in J} \rho_j$, where each $\rho_j$ is given by the generators of $D_n=\langle a, b \, | \, a^{n} =b^2= 1, bab = a^{-1} \rangle$ as follows: 
\begin{enumerate}[label=\emph{\roman*)}]
\item For $j\in J_0$
$$\left\lbrace\begin{array}{llll}
\rho_j(a):=(1,1) & \text{and} & \rho_j(b):=(1,-1) &  \text{if } f_j = \mx-1 \text{ and $\car(\Fqq)\neq 2$.} \\*[1mm]
\rho_j(a):=(-1,-1) & \text{and} & \rho_j(b):=(1,-1) &  \text{if } f_j = \mx+1 \text{ and $\car(\Fqq)\neq 2$.} \\*[1mm]
\rho_j(a):= \begin{bmatrix} 1&0 \\ 0&1\end{bmatrix} & \text{and} & \rho_j(b):=\begin{bmatrix} 0&1 \\ 1&0\end{bmatrix} & \text{if }  \text{$\car(\Fqq)=2$.}
 \end{array}\right.$$
 
\item For $j \in J_1$
$$
\rho_j(a):=
\left(\sigma_j \left( \begin{bmatrix} \alpha_j & 0 \\ 0 & \alpha_j^{-1} \end{bmatrix} \right) , \:\overline{\sigma_j} \left( \begin{bmatrix} \alpha_j^q & 0 \\ 0 & \alpha_j^{-q} \end{bmatrix} \right)  \right) 
 \quad 
\text{and} \quad  \rho_j(b):=\left(\sigma_j \left( \begin{bmatrix} 0&1  \\  1 &0 \end{bmatrix} \right),\: \overline{\sigma_j} \left( \begin{bmatrix} 0&1  \\ 1&0 \end{bmatrix}\right)\right),
$$
where the maps $\sigma_j$ and $\overline{\sigma_j}$ are those in \emph{Lemma~\ref{lemma:sigmas} (b)}. 

\item For $j\in J_2\cup J_3$ 
$$\rho_j(a):= \begin{bmatrix} \alpha_j & 0 \\ 0 & \alpha_j^{-1} \end{bmatrix} \quad \text{and} \quad  \rho_j(b):= \begin{bmatrix} 0&1  \\  1 &0 \end{bmatrix}.$$

\item For $j\in J_4$  
$$\rho_j(a):=\left(\begin{bmatrix} \alpha_j & 0 \\ 0 & \alpha_j^{-1} \end{bmatrix} ,\: \begin{bmatrix} \alpha_j^q & 0 \\ 0 & \alpha_j^{-q} \end{bmatrix}\right) \quad \text{and} \quad  \rho_j(b):=\left(\begin{bmatrix} 0&1  \\  1 &0 \end{bmatrix} , \: \begin{bmatrix} 0&1  \\  1 &0 \end{bmatrix}\right).
$$
\end{enumerate}
\end{theorem}

\begin{proof}
This is a direct application of \cite[Theorem 3.1]{Bro15}. For the sake of comprehensiveness, we prove that $\rho_j$ is well-defined for $j\in J_1$, \emph{i.e.} that its image belongs to $\mathcal{M}_2 \left( \Fqr{2}(\alpha_j + \alpha_j^{-1}) \right) \oplus \mathcal{M}_2 \left( \Fqr{2}(\alpha_j^q + \alpha_j^{-q}) \right)$. Observe that
\begin{eqnarray*}
\sigma_j\left(
		 \begin{bmatrix}
		\alpha_j & 0 \\ 0 & \alpha_j^{-1}
		\end{bmatrix}
		\right)&=& \begin{bmatrix}
1 & -\alpha_j \\ 1 & -\alpha_j^{-1}
\end{bmatrix}^{-1} \begin{bmatrix}
		\alpha_j & 0 \\ 0 & \alpha_j^{-1}
		\end{bmatrix}   \begin{bmatrix}
1 & -\alpha_j \\ 1 & -\alpha_j^{-1}
\end{bmatrix} \\ & = &
		\dfrac{1}{\alpha_j-\alpha_j^{-1}}\begin{bmatrix}
-\alpha_j^{-1} & \alpha_j \\ -1 & 1
\end{bmatrix}
		 \begin{bmatrix}
		\alpha_j & -\alpha_j^2 \\ \alpha_j^{-1} & -\alpha_j^{-2}
		\end{bmatrix} \\ & = & \begin{bmatrix}
		0 & 1 \\ -1 & \alpha_j + \alpha_j^{-1} 
		\end{bmatrix}\in \mathcal{M}_2 \left( \Fqr{2}(\alpha_j + \alpha_j^{-1}) \right),
\end{eqnarray*}
and similarly one can check
\[
\overline{\sigma_j}\left(
		 \begin{bmatrix}
		\alpha_j^q & 0 \\ 0 & \alpha_j^{-q}
		\end{bmatrix}
		\right)=		
		 \begin{bmatrix}
		0 & 1 \\ -1 & \alpha_j^q + \alpha_j^{-q}
		\end{bmatrix}
	\in \mathcal{M}_2 \left( \Fqr{2}(\alpha_j^q + \alpha_j^{-q}) \right).
\]
Further, we have 
\begin{eqnarray*}
\sigma_j\left(
		 \begin{bmatrix}
		0&1 \\ 1&0
		\end{bmatrix}
		\right)&=& \begin{bmatrix}
1 & -\alpha_j \\ 1 & -\alpha_j^{-1}
\end{bmatrix}^{-1} \begin{bmatrix}
		0&1 \\ 1&0
		\end{bmatrix}   \begin{bmatrix}
1 & -\alpha_j \\ 1 & -\alpha_j^{-1}
\end{bmatrix} \\ & = &
		\dfrac{1}{\alpha_j-\alpha_j^{-1}}\begin{bmatrix}
-\alpha_j^{-1} & \alpha_j \\ -1 & 1
\end{bmatrix}
		 \begin{bmatrix}
		1& -\alpha_j^{-1} \\ 1 & -\alpha_j
		\end{bmatrix} \\ & = & \begin{bmatrix}
		1 & -(\alpha_j+\alpha_j^{-1}) \\ 0 & -1
		\end{bmatrix}\in \mathcal{M}_2 \left( \Fqr{2}(\alpha_j + \alpha_j^{-1}) \right),
\end{eqnarray*}
and similarly 
$$\overline{\sigma_j}\left(
		 \begin{bmatrix}
		0 & 1 \\ 1 & 0
		\end{bmatrix}
		\right)=
		 \begin{bmatrix}
		1 & -(\alpha_j^q+\alpha_j^{-q}) \\ 0 & -1
		\end{bmatrix}
		\in \mathcal{M}_2 \left( \Fqr{2}(\alpha_j^q + \alpha_j^{-q}) \right).$$
\end{proof}

\begin{remark}
\label{remark_fields}
Assume Hypotheses~\ref{hipotesis_descomposicion_Fq2Dn}. By Remark \ref{arrels_recipr}, it holds $\Fqr{2}(\alpha_j+\alpha_j^{-1})\cong \Fqr{r_j}\cong\Fqr{2}(\alpha_j^q+\alpha_j^{-q})$ for $j\in J_1$, and $\Fqr{2}(\alpha_j)\cong \Fqr{2r_j}\cong \Fqr{2}(\alpha_j^q)$ for $j\in J\smallsetminus J_1$.
\end{remark}


\subsection{On the Hermitian dual of a \texorpdfstring{$D_n$-code}{Dn-code}}\label{DualHerDn}

Given a dihedral code $\cC$ over $\Fqr{2}$, whenever $\car(\Fqq)\nmid n$, the algebra isomorphism $\rho$ of Theorem~\ref{theo_dih_refined} allows us to obtain a decomposition of $I_{\cC}\subseteq \Fqr{2}[D_n]$ as a  sum of principal ideals of matrix rings over finite fields, that is $\rho(I_{\cC})=\bigoplus_{j \in J}  I_j\subseteq \bigoplus_{j \in J} A_j$. Moreover, since the corresponding decomposition of the Euclidean dual $(I_{\cC})^{\perp_E}$ is already given in Theorem~\ref{dih_Euclidean_dual}, we are going to take into account (\ref{eq:orthogonal_ideal_Hermitian}) to achieve the appropriate decomposition of $(I_{\cC})^{\perp_H}$, that is, $\rho((I_{\cC})^{\perp_H})$. In other words, our goal is to find $\varphi$ such that the following diagram is commutative.

\begin{equation}
\label{diagrama}
\begin{tikzcd}[row sep=large, column sep=large]
\Fqr{2}[D_n] \arrow[d, "\rho"'] \arrow[r, "\operatorname{conj}_q"] & \Fqr{2}[D_n] \arrow[d, "\rho"]     \\
\displaystyle\bigoplus_{j \in J} A_j \arrow[r, "\varphi"]  & \displaystyle\bigoplus_{j \in J} A_j
\end{tikzcd}
\end{equation}
In this way, we can find $\rho((I_{\cC})^{\perp_H})$ through $\varphi$:
\[
\rho((I_{\cC})^{\perp_H})=(\rho \circ \operatorname{conj}_q)((I_{\cC})^{\perp_E})=\varphi(\rho((I_{\cC})^{\perp_E})).
\]

In order to define such $\varphi$, we proceed as explained in $(\ref{reduccioDual})$, so that we reduce the problem of finding the dual of $ I_{\cC}$ to that of finding the dual of its summands $I_{j}$ in the respective rings $A_j$. To do so, we need first to introduce the next two maps. Observe that for $j\in J_{1}\cup\dots \cup J_{4}$, the possible matrices appearing in $A_{j}$ take values in one of the following four fields: $\Fqr{2}(\alpha_j)$, $\Fqr{2}(\alpha_j^{q})$, $\Fqr{2}(\alpha_j+\alpha_j^{-1})$ and $\Fqr{2}(\alpha_j^q+\alpha_j^{-q})$. Hence, for any positive integer $m$, we define the following maps $\mathcal{H}^m$ and $\mathcal{T}^m$ for square matrices over any one of these four fields.

\begin{equation}\label{HiT}
\mathcal{H}^m : \begin{bmatrix}
x_1 & x_2 \\ x_3 & x_4
\end{bmatrix} \mapsto \begin{bmatrix}
x_1^{q^m} & x_2^{q^m} \\ x_3^{q^m} & x_4^{q^m}
\end{bmatrix}
\hspace{1.5cm}
\mathcal{T}^m : \begin{bmatrix}
x_1 & x_2 \\ x_3 & x_4
\end{bmatrix} \mapsto \begin{bmatrix}
x_4^{q^m} & x_3^{q^m} \\ x_2^{q^m} & x_1^{q^m}
\end{bmatrix}
\end{equation}

Now, taking into account Remark \ref{remark_fields}, we have the following. 
\begin{lemma}
\label{automorphisms}
Assume \emph{Hypotheses~\ref{hipotesis_descomposicion_Fq2Dn}}. The maps $\mathcal{H}^m$ and $\mathcal{T}^m$ defined in \emph{(\ref{HiT})} are automorphisms. In fact
\[
(\mathcal{H}^m)^{-1}=\mathcal{H}^{r_j-m} \mbox{ and } (\mathcal{T}^m)^{-1}=\mathcal{T}^{r_j-m}, \mbox{ for } j\in J_1 \mbox{ and } 1\leq m\leq r_j,
\]
and
\[
(\mathcal{H}^m)^{-1}=\mathcal{H}^{2r_j-m} \mbox{ and } (\mathcal{T}^m)^{-1}=\mathcal{T}^{2r_j-m}, \mbox{ for } j\in J_{2}\cup J_3\cup J_{4} \mbox{ and } 1\leq m\leq 2r_j.
\]
\end{lemma}

\begin{proof}
Take $j\in J_1$. Note that the inverse of $\mathcal{H}^m$ is well-defined, and indeed it is $\mathcal{H}^{r_j-m}$:
\begin{eqnarray*}\left(\mathcal{H}^m \circ \mathcal{H}^{r_j-m}\right) \left( \begin{bmatrix} x_1 & x_2 \\ x_3 & x_4 \end{bmatrix} \right) 
&=& \mathcal{H}^{m} \left( \begin{bmatrix} x_1^{q^{r_j-m}} & x_2^{q^{r_j-m}} \\ x_3^{q^{r_j-m}} & x_4^{q^{2s_j-m}} \end{bmatrix} \right) \\
&=& \begin{bmatrix} x_1^{q^{r_j}} & x_2^{q^{r_j}} \\ x_3^{q^{r_j}} & x_4^{q^{r_j}} \end{bmatrix} 
= \begin{bmatrix} x_1 & x_2 \\ x_3 & x_4 \end{bmatrix}. 
\end{eqnarray*}
In all other cases, we obtain the inverse maps of the statement in a similar manner.  Finally, it is straightforward to check that $\mathcal{H}^m$ and $\mathcal{T}^m$ are homomorphisms. 
\end{proof}

\medskip

\begin{theorem}
\label{theo_commutative}
Assume \emph{Hypotheses~\ref{hipotesis_descomposicion_Fq2Dn}}, and consider the algebra isomorphism $\rho$ of \emph{Theorem~\ref{theo_dih_refined}}. Then the  diagram \emph{(\ref{diagrama})} is commutative if 
$$
\varphi = \bigoplus_{j \in J} \varphi_j , \mbox{ where } \varphi_{j}:A_{j} \rightarrow A_{j} \mbox{ is given by }
$$

$$
\begin{cases}
\varphi_j(x , y) := (x^q , \: y^q) & \text{if } j \in J_0 \text{ and $\car(\Fqq)\neq 2$} \\

\varphi_j(X) := \mathcal{H}^1(X) & \text{if } j \in J_0 \text{ and $\car(\Fqq)=2$} \\

\varphi_j(X , Y) := \left( \left(\sigma_j \circ \mathcal{H}^{2r_j-1} \circ \overline{\sigma_j}^{-1} \right)(Y) , \: \left( \overline{\sigma_j} \circ \mathcal{H}^{1} \circ \sigma_j^{-1} \right)(X) \right) & \text{if } j \in J_1 \\

\varphi_j(X) := \mathcal{H}^{r_j}(X) & \text{if } j \in J_2 \\

\varphi_j(X) := \mathcal{T}^{r_j}(X) & \text{if } j \in J_3 \\

\varphi_j(X , Y) := \left(\mathcal{H}^{2r_j-1}(Y) , \: \mathcal{H}^{1}(X)\right) & \text{if } j \in J_4 \\
\end{cases},$$ 

\smallskip 

\noindent where $\sigma_j$ and $\overline{\sigma_j}$ are the maps in \emph{Lemma~\ref{lemma:sigmas} (b)}, and $\mathcal{H}^m$ and $\mathcal{T}^m$ are the maps in \emph{(\ref{HiT})}.
\end{theorem}

\begin{proof}
Consider the presentation $D_n=\langle a, b \, | \, a^{n} =b^2= 1, bab = a^{-1} \rangle$. Since $$D_n=\{1,a,a^2,...,a^{n-1}, b,ab, a^2b, ... , a^{n-1}b\},$$ any element $u \in \Fqr{2}[D_n]$ can be written as $P(a) + Q(a)b$ for certain polynomials $P, Q\in \Fqr{2}[\mx]$ of degree at most $n-1$. Then
\[
\operatorname{conj}_{q}(u)=\operatorname{conj}_{q}(P(a) + Q(a)b)=\overline{P}(a) + \overline{Q}(a)b,
\]
where $\operatorname{conj}_q$ is the map defined in (\ref{defConj}). We aim to show that the map $\varphi$ defined in the statement verifies that $\varphi \circ \rho=\rho \circ \operatorname{conj}_{q}$, that is,  $\rho(\operatorname{conj}_{q}(u)) = \varphi(\rho(u))$, for all $u \in \Fqr{2}[D_n]$. Let us argue case by case.

\medskip

\noindent \underline{$j \in J_0$ and $\car(\Fqq)\neq 2$.} If $f_j = \mx-1$, then by Lemma \ref{propiedades_reci_conj}(a):
\begin{align*}
\rho_j(\operatorname{conj}_{q}(u)) &= \rho_j\left(\overline{P}(a) + \overline{Q}(a)b \right) \\ 
						     &=\overline{P}(\rho_j(a))+ \overline{Q}(\rho_j(a))\rho_j(b) \\
						     &= \left(\overline{P}(1)+\overline{Q}(1), \overline{P}(1)-\overline{Q}(1)\right) \\ 
   					             &= \left(P(1)^q+Q(1)^q , P(1)^q-Q(1)^q\right)
\end{align*}
\begin{align*}
\varphi_j(\rho_j(u)) &= \varphi_j \big(P(1)+Q(1) , P(1)-Q(1) \big) \\ 
                               &= \left(P(1)^q+Q(1)^q , P(1)^q-Q(1)^q\right)
\end{align*}

If $f_j = \mx+1$, then similarly:
\begin{align*}
\rho_j(\operatorname{conj}_{q}(u)) &= \rho_j\left(\overline{P}(a) + \overline{Q}(a)b\right) \\ 
                                                        &=\overline{P}(\rho_j(a))+ \overline{Q}(\rho_j(a))\rho_j(b) \\
                                                      &= \left(\overline{P}(-1)+\overline{Q}(-1), \overline{P}(-1)-\overline{Q}(-1)\right) \\ 
                                                  &= \left(P(-1)^q+Q(-1)^q , P(-1)^q-Q(-1)^q\right)
\end{align*}
\begin{align*}
\varphi_j(\rho_j(u)) &= \varphi_j \big(P(-1)+Q(-1) , P(-1)-Q(-1) \big) \\ 
                              &= \left(P(-1)^q+Q(-1)^q ,  P(-1)^q-Q(-1)^q\right)
\end{align*}
\noindent \underline{$j \in J_0$ and $\car(\Fqq)=2$.} In this case we obtain:
\begin{center}
\begin{minipage}{0.45\textwidth}
\begin{align*}
\rho_j(\operatorname{conj}_{q}(u)) &= \rho_j\left(\overline{P}(a) + \overline{Q}(a)b \right) \\ &= \begin{bmatrix} \overline{P}(1)&\overline{Q}(1)\\\overline{Q}(1)&\overline{P}(1)  \end{bmatrix} \\ &=  \begin{bmatrix} P(1)^q& Q(1)^q\\ Q(1)^q& P(1)^q  \end{bmatrix} 
\end{align*}
\end{minipage}
\begin{minipage}{0.45\textwidth}
\begin{align*}
\varphi_j(\rho_j(u)) &= \varphi_j \left(\begin{bmatrix} P(1)& Q(1)\\ Q(1)& P(1)  \end{bmatrix}  \right) \\ &= \begin{bmatrix} P(1)^q& Q(1)^q\\ Q(1)^q& P(1)^q  \end{bmatrix} 
\end{align*}
\end{minipage}
\end{center}

\noindent \underline{$j \in J_1$.} Note that $\alpha_j^{-1}$ is also a root of $f_j$, so applying Lemma~\ref{lemma_tech_rj} (a) and Lemma~\ref{propiedades_reci_conj} (a) we get:
\begin{align*}
\rho_j(\operatorname{conj}_{q}(u)) &= \rho_j\left(\overline{P}(a) + \overline{Q}(a)b\right) \\
&= \left(\sigma_j \left( \begin{bmatrix} \overline{P}(\alpha_j) & \overline{Q}(\alpha_j) \\ \overline{Q}(\alpha_j^{-1}) & \overline{P}(\alpha_j^{-1}) \end{bmatrix} \right) , \overline{\sigma_j} \left( \begin{bmatrix} \overline{P}(\alpha_j^q) & \overline{Q}(\alpha_j^q) \\ \overline{Q}(\alpha_j^{-q}) & \overline{P}(\alpha_j^{-q}) \end{bmatrix} \right)\right) \\
&= \left(\sigma_j \left( \begin{bmatrix} P(\alpha_j^{q})^{q^{2r_j-1}} & Q(\alpha_j^{q})^{q^{2r_j-1}} \\ Q(\alpha_j^{-q})^{q^{2r_j-1}} & P(\alpha_j^{-q})^{q^{2r_j-1}} \end{bmatrix} \right) , \overline{\sigma_j} \left( \begin{bmatrix} P(\alpha_j)^q & Q(\alpha_j)^q \\ Q(\alpha_j^{-1})^q & P(\alpha_j^{-1})^q \end{bmatrix} \right) \right)
\end{align*}

\begin{align*}
\varphi_j(\rho_j(u)) &= \varphi_j \left( \sigma_j \left( \begin{bmatrix} P(\alpha_j) & Q(\alpha_j) \\ Q(\alpha_j^{-1}) & P(\alpha_j^{-1}) \end{bmatrix} \right) , \overline{\sigma_j} \left( \begin{bmatrix} P(\alpha_j^q) & Q(\alpha_j^q) \\ Q(\alpha_j^{-q}) & P(\alpha_j^{-q}) \end{bmatrix} \right) \right)\\
&= \left(\left( \sigma_j \circ \mathcal{H}^{2r_j-1} \right) \left( \begin{bmatrix} P(\alpha_j^q) & Q(\alpha_j^q) \\ Q(\alpha_j^{-q}) & P(\alpha_j^{-q}) \end{bmatrix} \right) , \left( \overline{\sigma_j} \circ \mathcal{H}^{1} \right) \left( \begin{bmatrix} P(\alpha_j) & Q(\alpha_j) \\ Q(\alpha_j^{-1}) & P(\alpha_j^{-1}) \end{bmatrix} \right) \right) \\
&= \left(\sigma_j \left( \begin{bmatrix} P(\alpha_j^q)^{q^{2r_j-1}} & Q(\alpha_j^q)^{q^{2r_j-1}} \\ Q(\alpha_j^{-q})^{q^{2r_j-1}} & P(\alpha_j^{-q})^{q^{2r_j-1}} \end{bmatrix} \right) , \overline{\sigma_j} \left( \begin{bmatrix} P(\alpha_j)^q & Q(\alpha_j)^q \\ Q(\alpha_j^{-1})^q & P(\alpha_j^{-1})^q \end{bmatrix} \right)\right)
\end{align*}

\noindent \underline{$j \in J_2$.} In this case we apply Lemma~\ref{lemma_tech_rj} (b) and we get:

\begin{minipage}{0.45\textwidth}
\begin{align*}
\rho_j(\operatorname{conj}_{q}(u)) &= \rho_j\left(\overline{P}(a) + \overline{Q}(a)b\right) \\
&= \begin{bmatrix} \overline{P}(\alpha_j) & \overline{Q}(\alpha_j) \\ \overline{Q}(\alpha_j^{-1}) & \overline{P}(\alpha_j^{-1}) \end{bmatrix} \\
&= \begin{bmatrix} P(\alpha_j)^{q^{r_j}} & Q(\alpha_j)^{q^{r_j}} \\ Q(\alpha_j^{-1})^{q^{r_j}} & P(\alpha_j^{-1})^{q^{r_j}} \end{bmatrix}
\end{align*}
\end{minipage}
\begin{minipage}{0.45\textwidth}
\begin{align*}
\varphi_j(\rho_j(u)) &= \varphi_j \left( \begin{bmatrix} P(\alpha_j) & Q(\alpha_j) \\ Q(\alpha_j^{-1}) & P(\alpha_j^{-1}) \end{bmatrix} \right)\\
&= \mathcal{H}^{r_j} \left( \begin{bmatrix} P(\alpha_j) & Q(\alpha_j) \\ Q(\alpha_j^{-1}) & P(\alpha_j^{-1}) \end{bmatrix} \right) \\
&= \begin{bmatrix} P(\alpha_j)^{q^{r_j}} & Q(\alpha_j)^{q^{r_j}} \\ Q(\alpha_j^{-1})^{q^{r_j}} & P(\alpha_j^{-1})^{q^{r_j}} \end{bmatrix}
\end{align*}
\end{minipage}

\medskip

\noindent \underline{$j \in J_3$.} Using Lemma~\ref{lemma_tech_rj} (c) it follows:

\begin{minipage}{0.45\textwidth}
\begin{align*}
\rho_j(\operatorname{conj}_{q}(u)) &= \rho_j\left(\overline{P}(a) + \overline{Q}(a)b\right) \\
&= \begin{bmatrix} \overline{P}(\alpha_j) & \overline{Q}(\alpha_j) \\ \overline{Q}(\alpha_j^{-1}) & \overline{P}(\alpha_j^{-1}) \end{bmatrix} \\
&= \begin{bmatrix} P(\alpha_j^{-1})^{q^{r_j}} & Q(\alpha_j^{-1})^{q^{r_j}} \\ Q(\alpha_j)^{q^{r_j}} & P(\alpha_j)^{q^{r_j}} \end{bmatrix}
\end{align*}
\end{minipage}
\begin{minipage}{0.45\textwidth}
\begin{align*}
\varphi_j(\rho_j(u)) &= \varphi_j \left( \begin{bmatrix} P(\alpha_j) & Q(\alpha_j) \\ Q(\alpha_j^{-1}) & P(\alpha_j^{-1}) \end{bmatrix} \right)\\
&= \mathcal{T}^{r_j} \left( \begin{bmatrix} P(\alpha_j) & Q(\alpha_j) \\ Q(\alpha_j^{-1}) & P(\alpha_j^{-1}) \end{bmatrix} \right) \\
&= \begin{bmatrix} P(\alpha_j^{-1})^{q^{r_j}} & Q(\alpha_j^{-1})^{q^{r_j}} \\ Q(\alpha_j)^{q^{r_j}} & P(\alpha_j)^{q^{r_j}} \end{bmatrix}
\end{align*}
\end{minipage}

\medskip

\noindent \underline{$j \in J_4$.} Here we can use Lemma~\ref{lemma_tech_rj} (a) and Lemma~\ref{propiedades_reci_conj} (a) again to get:

\begin{align*}
\rho_j(\operatorname{conj}_{q}(u)) &= \rho_j\left(\overline{P}(a) + \overline{Q}(a)b\right) \\
&=\left( \begin{bmatrix} \overline{P}(\alpha_j) & \overline{Q}(\alpha_j) \\ \overline{Q}(\alpha_j^{-1}) & \overline{P}(\alpha_j^{-1}) \end{bmatrix} , \begin{bmatrix} \overline{P}(\alpha_j^q) & \overline{Q}(\alpha_j^q) \\ \overline{Q}(\alpha_j^{-q}) & \overline{P}(\alpha_j^{-q}) \end{bmatrix} \right) \\
&= \left( \begin{bmatrix} P(\alpha_j^{q})^{q^{2r_j-1}} & Q(\alpha_j^{q})^{q^{2r_j-1}} \\ Q(\alpha_j^{-q})^{q^{2r_j-1}} & P(\alpha_j^{-q})^{q^{2r_j-1}} \end{bmatrix} , \begin{bmatrix} P(\alpha_j)^q & Q(\alpha_j)^q \\ Q(\alpha_j^{-1})^q & P(\alpha_j^{-1})^q \end{bmatrix}\right)
\end{align*}

\begin{align*}
\varphi_j(\rho_j(u)) &= \varphi_j \left( \begin{bmatrix} P(\alpha_j) & Q(\alpha_j) \\ Q(\alpha_j^{-1}) & P(\alpha_j^{-1}) \end{bmatrix} , \begin{bmatrix} P(\alpha_j^q) & Q(\alpha_j^q) \\ Q(\alpha_j^{-q}) & P(\alpha_j^{-q}) \end{bmatrix} \right)\\
&= \left( \mathcal{H}^{2r_j-1} \left( \begin{bmatrix} P(\alpha_j^q) & Q(\alpha_j^q) \\ Q(\alpha_j^{-q}) & P(\alpha_j^{-q}) \end{bmatrix} \right) , \mathcal{H}^{1} \left( \begin{bmatrix} P(\alpha_j) & Q(\alpha_j) \\ Q(\alpha_j^{-1}) & P(\alpha_j^{-1}) \end{bmatrix} \right)\right) \\
&= \left(\begin{bmatrix} P(\alpha_j^q)^{q^{2r_j-1}} & Q(\alpha_j^q)^{q^{2r_j-1}} \\ Q(\alpha_j^{-q})^{q^{2r_j-1}} & P(\alpha_j^{-q})^{q^{2r_j-1}} \end{bmatrix} , \begin{bmatrix} P(\alpha_j)^q & Q(\alpha_j)^q \\ Q(\alpha_j^{-1})^q & P(\alpha_j^{-1})^q \end{bmatrix}\right).
\end{align*}
\vspace*{-2mm}
\end{proof}

\medskip

As mentioned at the beginning of the subsection, in order to obtain the Hermitian dual of a $D_n$-code $\cal C$, it is now enough to apply the map $\varphi$ in the above result to $\rho((I_{\cC})^{\perp_{e}})$. We do this in the following theorem, which is the main result of this section.

\begin{theorem}
\label{dih_hermitic_orthogonal}
Assume \emph{Hypotheses~\ref{hipotesis_descomposicion_Fq2Dn}}. Let $\mathcal{C}$ be a $D_n$-code over $\Fqq$, and let $I_{\cC}$ be its corresponding ideal in $\Fqq[D_n]$ expressed, using the isomorphism $\rho$ given in \emph{Theorem~\ref{theo_dih_refined}}, as 
$$
I_{\cC}\overset{\rho}{\cong} \displaystyle\bigoplus_{j\in J} I_j\subseteq \bigoplus_{j\in J}A_{j}.
$$ 
Set $I_j:=\langle X\rangle \oplus \langle Y \rangle$ for $j\in J_1\cup J_4$, and $I_j:=\langle X\rangle$ for $j\in J_2\cup J_3$, where $X,Y$ are square matrices in row reduced echelon form over the appropriate fields stated in \emph{Theorem~\ref{theo_dih_refined}}. Then $I_{\cC^{\perp_H}}\subseteq \Fqq[D_n]$ verifies that 
$$
I_{\cC^{\perp_H}}\overset{\rho}{\cong} \displaystyle\bigoplus_{j\in J} {I}_j^{\perp_H} \subseteq \bigoplus_{j\in J}A_{j},
$$
where:
\begin{enumerate}
\item[\emph{i)}] for $j\in J_0$ and $\car(\Fqq)\neq 2$:
$$
{I}_j^{\perp_H}=\begin{cases}
\Fqr{2}\oplus \Fqr{2} & \text{if } I_j=\pmb{0}\oplus\pmb{0} \\
\pmb{0}\oplus \Fqr{2} & \text{if } I_j=\Fqr{2}\oplus\pmb{0} \\
\Fqr{2}\oplus \pmb{0} & \text{if } I_j=\pmb{0}\oplus\Fqr{2} \\
\pmb{0}\oplus\pmb{0} & \text{if } I_j=\Fqr{2}\oplus \Fqr{2} \\
\end{cases}$$

\item[\emph{ii)}] for $j\in J_0$ and $\car(\Fqq)=2$:
$$
{I}_j^{\perp_H}=\begin{cases}
\Fqr{2}[C_2] & \text{if } I_j=\pmb{0} \\*[2mm]
\left\langle \begin{bmatrix}1&1\\1&1\end{bmatrix}\right\rangle_{\Fqq[C_{2}]} & \text{if } I_j=\left\langle \begin{bmatrix}1&1\\1&1\end{bmatrix}\right\rangle_{\Fqq[C_{2}]}  \\*[5mm]
\pmb{0} & \text{if } I_j=\Fqr{2}[C_2] \\
\end{cases}$$

\item[\emph{iii)}] for $j\in J_1$ it holds ${I}_j ^{\perp_H}=  \langle \widetilde{Y} \rangle \oplus \langle \widetilde{X} \rangle$ with one of the following options for $\widetilde{X}$ and $\widetilde{Y}$: 
$$
\widetilde{X}:=\begin{cases}  
I-X & \text{if } X=\pmb{0} \text{ or } X=I\\*[3mm]  \begin{bmatrix} 2& -(\alpha_j+\alpha_j^{-1})^q\\0&0 \end{bmatrix} & \text{if } X=\begin{bmatrix} 0& 1\\0&0 \end{bmatrix}\\*[5mm] 
 \begin{bmatrix} (2\lambda_j+\alpha_j+\alpha_j^{-1})^q& (-2-(\alpha_j+\alpha_j^{-1})\lambda_j)^q\\0&0 \end{bmatrix} & \text{if } X=\begin{bmatrix} 1& \lambda_j\\0&0 \end{bmatrix}, \; \lambda_j\in  \Fqr{2}(\alpha_j+\alpha_j^{-1})
 \end{cases}$$
$$
\widetilde{Y}:=\begin{cases}
I-Y & \text{if } Y=\pmb{0} \text{ or } Y=I\\*[3mm]  \begin{bmatrix} 2& -(\alpha_j+\alpha_j^{-1})\\0&0 \end{bmatrix} & \text{if } Y=\begin{bmatrix} 0& 1\\0&0 \end{bmatrix}\\*[5mm] 
 \begin{bmatrix} 2\lambda_j^{q^{r_j-1}}+\alpha_j+\alpha_j^{-1}& -2-(\alpha_j+\alpha_j^{-1})\lambda_j^{q^{r_j-1}}\\0&0 \end{bmatrix} & \text{if } Y=\begin{bmatrix} 1& \lambda_j\\0&0 \end{bmatrix}, \; \lambda_j\in  \Fqr{2}(\alpha_j^q+\alpha_j^{-q})
\end{cases}
$$

\item[\emph{iv)}] for $j\in J_2$ it holds  ${I}_j ^{\perp_H}=  \langle \widetilde{X} \rangle$ with one of the following options: 
$$
\widetilde{X}:=\begin{cases}
I-X & \text{if } X=\pmb{0} \text{ or } X=I\\*[3mm]
\begin{bmatrix} 0& 1\\0&0 \end{bmatrix} & \text{if } X=\begin{bmatrix} 0& 1\\0&0 \end{bmatrix} \\*[5mm] 
\begin{bmatrix} 1 & -\lambda_j^{q^{r_j}}\\ 0&0 \end{bmatrix} & \text{if } X=\begin{bmatrix} 1& \lambda_j\\0&0 \end{bmatrix}, \; \lambda_j\in \Fqr{2}(\alpha_j)
\end{cases}
$$ 

\item[\emph{v)}] for $j\in J_3$ it holds  ${I}_j ^{\perp_H}=  \langle \widetilde{X} \rangle$ with one of the following options: 
$$
\widetilde{X}:=\begin{cases}
I-X & \text{if } X=\pmb{0} \text{ or } X=I\\*[3mm]
\begin{bmatrix} 1&0\\0&0 \end{bmatrix} & \text{if } X=\begin{bmatrix} 0& 1\\0&0 \end{bmatrix} \\*[5mm] 
 \begin{bmatrix} -\lambda_j^{q^{r_j}}&1\\ 0&0 \end{bmatrix} & \text{if } X=\begin{bmatrix} 1& \lambda_j\\0&0 \end{bmatrix}, \; \lambda_j\in  \Fqr{2}(\alpha_j)
\end{cases}
$$ 

\item[\emph{vi)}] for $j\in J_4$ it holds   ${I}_j^{\perp_H} =  \langle \widetilde{Y} \rangle \oplus \langle \widetilde{X} \rangle$  with one of the following options for $\widetilde{X}$ and $\widetilde{Y}$: 
$$\widetilde{X}:=\begin{cases}  
I-X & \text{if } X=\pmb{0} \text{ or } X=I\\*[3mm] 
 \begin{bmatrix} 0& 1\\0&0 \end{bmatrix} & \text{if } X=\begin{bmatrix} 0& 1\\0&0 \end{bmatrix}\\*[5mm] 
 \begin{bmatrix} 1 & -\lambda_j^q\\0&0 \end{bmatrix} & \text{if } X=\begin{bmatrix} 1& \lambda_j\\0&0 \end{bmatrix}, \; \lambda_j\in  \Fqr{2}(\alpha_j)\end{cases}$$
$$\widetilde{Y}:=\begin{cases}
I-Y & \text{if } Y=\pmb{0} \text{ or } Y=I\\*[3mm]  \begin{bmatrix} 0& 1\\0&0 \end{bmatrix} & \text{if } Y=\begin{bmatrix} 0& 1\\0&0 \end{bmatrix}\\*[5mm] 
 \begin{bmatrix} 1 & -\lambda_j^{q^{2r_j-1}}\\0&0 \end{bmatrix} & \text{if } Y=\begin{bmatrix} 1& \lambda_j\\0&0 \end{bmatrix}, \; \lambda_j\in  \Fqr{2}(\alpha_j^q)
\end{cases}
$$ 
\end{enumerate}
\end{theorem}

\begin{proof}
Let us consider the diagram (\ref{diagrama}), which is commutative by Theorem~\ref{theo_commutative}. Since $I_{\cC^{\perp_H}}=\operatorname{conj}_q(I_{\mathcal{C}^{\perp_E}})$ (see (\ref{eq:orthogonal_ideal_Hermitian})), then it is enough to apply the map $\varphi$ in Theorem~\ref{theo_commutative} to  $\rho(I_{\mathcal{C}^{\perp_E}})$ in order to get the desired conclusions. In virtue of $(\ref{reduccioDual})$ the problem is therefore reduced to apply $\varphi_j$ to $I_j^{\perp_E}$ for any $j\in J$. Thus, we proceed case by case.

\medskip

i) For $j\in J_{0}$ and $\car(\Fqq)\neq 2$, take an ideal $z_{j_1}\Fqr{2}\oplus z_{j_2}\Fqr{2}$, where $z_{j_1},z_{j_2}\in \{0,1\}$. Its Euclidean dual is $(1-z_{j_1})\mathbb{F}_{q^2}\oplus (1-z_{j_2})\mathbb{F}_{q^2}$ (see Theorem~\ref{dih_Euclidean_dual}~(i)). Now, since $\varphi_j(x, y)=(x^q, y^q)$, for any $(x,y)\in \Fqr{2}\oplus \Fqr{2}$, we easily get that the Hermitian dual is $(1-z_{j_1})\mathbb{F}_{q^2}\oplus (1-z_{j_2})\mathbb{F}_{q^2}$.

\medskip

ii) For $j\in J_{0}$ and $\car(\Fqq)= 2$, it is enough to apply $\varphi_j=\mathcal{H}^1$ to the Euclidean dualities that appear in  Theorem~\ref{dih_Euclidean_dual}~(ii), for each ideal of $\Fqr{2}[C_2]$.

\medskip

iii)  For $j\in J_{1}$, notice that $f_{j}$ and $\overline{f_j}$ are both self-reciprocal and different polynomials. Thus, we will first apply Theorem~\ref{dih_Euclidean_dual}~(iii) for both $\langle X\rangle$ and $\langle Y\rangle $ in order to obtain their Euclidean duals.
If $X\in \{\pmb{0},I\}$, it holds $\langle X\rangle^{\perp_E}=\langle I-X\rangle$. In view of $\varphi_j$ in Theorem~\ref{theo_commutative}, it is easy to see now that $\overline{\sigma_j}\circ \mathcal{H}^1\circ \sigma_j^{-1}$, which is an automorphism of $A_j$ by Lemma~\ref{automorphisms}, fixes $I-X$, so $\widetilde{X}=I-X$. An analogous argument also shows that $\widetilde{Y}=I-Y$ whenever $Y\in \{\pmb{0},I\}$.

Let us suppose now that $X =\left[\begin{smallmatrix} 0&1\\0&0\end{smallmatrix}\right]$, so in virtue of  Theorem~\ref{dih_Euclidean_dual}~(iii) we get $$\langle X\rangle^{\perp_E}=\left\langle\begin{bmatrix} 2 & -\alpha_j-\alpha_j^{-1}\\0&0\end{bmatrix}\right\rangle.$$ 
Set $M_j$ for this last generator matrix. Applying $\varphi_j$, and the definition of the automorphism $\mathcal{H}^1$, we deduce
\begin{eqnarray*} 
(\overline{\sigma_j}\circ \mathcal{H}^1\circ \sigma_j^{-1} )(M_j) &=&  \overline{\sigma_j}\circ \mathcal{H}^1(Z_jM_jZ_j^{-1}) \\ 
&=& \overline{\sigma_j} (\mathcal{H}^1(Z_j)\mathcal{H}^1(M_j)\mathcal{H}^1(Z_j^{-1})) \\
& = & \overline{Z}_j^{-1}\mathcal{H}^1(Z_j)\mathcal{H}^1(M_j)\mathcal{H}^1(Z_j)^{-1}\overline{Z}_j \\
& = & \overline{Z}_j^{-1}\overline{Z}_j\mathcal{H}^1(M_j)\overline{Z}_j^{-1}\overline{Z}_j \\
& = & \mathcal{H}^1(M_j) \\
& = & \begin{bmatrix} 2& -(\alpha_j+\alpha_j^{-1})^q\\0&0 \end{bmatrix}.
\end{eqnarray*}
A similar reasonament also proves the remaining case, where $X = \left[ \begin{smallmatrix}1 & \lambda_j\\ 0&0\end{smallmatrix}\right]$ with $\lambda_j\in\Fqr{2}(\alpha_j+\alpha_j^{-1})$.

We finally show the case where $ Y = \left[\begin{smallmatrix} 1 & \lambda_j\\0&0\end{smallmatrix}\right]$ with $\lambda_j\in\Fqr{2}(\alpha_j^q+\alpha_j^{-q})$, since the other one can be analogously argued as before. Since $Y \in \mathcal{M}_2\left(\Fqr{2}(\alpha_j^q+\alpha_j^{-q})\right)$, which corresponds to the self-reciprocal polynomial $\overline{f_j}$, and $\alpha_j^q$ is a root of $\overline{f_j}$, then by Theorem~\ref{dih_Euclidean_dual}~(iii) we get $$\langle Y\rangle^{\perp_E}=\left\langle\begin{bmatrix}  \alpha_j^q+\alpha_j^{-q}+2\lambda_j &-2-(\alpha_j^q+\alpha_j^{-q})\lambda_j \\ 0&0 \end{bmatrix}\right\rangle.$$ Set $M_j$ for this last generator matrix. Now applying $\varphi_j$ in Theorem~\ref{theo_commutative} we obtain 
\begin{eqnarray*} 
(\sigma_j\circ \mathcal{H}^{2r_j-1}\circ \overline{\sigma_j}^{\,-1}) (M_j) &=&  \sigma_j\circ \mathcal{H}^{2r_j-1}(\overline{Z_j}M_j\overline{Z_j}^{\,-1}) \\ 
&=& \sigma_j (\mathcal{H}^{2r_j-1}(\overline{Z_j})\mathcal{H}^{2r_j-1}(M_j)\mathcal{H}^{2r_j-1}(\overline{Z_j}^{\,-1})) \\
& = & Z_j^{-1}\mathcal{H}^{2r_j-1}(\overline{Z_j})\mathcal{H}^{2r_j-1}(M_j)\mathcal{H}^{2r_j-1}(\overline{Z_j}^{\,-1})Z_j \\
& = & Z_j^{-1}Z_j\mathcal{H}^{2r_j-1}(M_j)Z_j^{\,-1}Z_j \\
& = & \mathcal{H}^{2r_j-1}(M_j) \\
& = &  \begin{bmatrix} (\alpha_j^q+\alpha_j^{-q}+2\lambda_j)^{q^{2r_j-1}} & (-2-(\alpha_j^q+\alpha_j^{-q})\lambda_j)^{q^{2r_j-1}} \\ 0&0 \end{bmatrix} \\
& = &  \begin{bmatrix} \alpha_j^{q^{2r_j}}+(\alpha_j^{-1})^{q^{2r_j}}+2\lambda_j^{q^{2r_j-1}} & -2-(\alpha_j^{q^{2r_j}}+(\alpha_j^{-1})^{q^{2r_j}})\lambda_j^{q^{2r_j-1}} \\ 0&0 \end{bmatrix} \\
& =&  \begin{bmatrix} \alpha_j+\alpha_j^{-1}+2\lambda_j^{q^{r_j-1}} & -2-(\alpha_j+\alpha_j^{-1})\lambda_j^{q^{r_j-1}} \\ 0&0 \end{bmatrix},
\end{eqnarray*}
where in the last equality we have utilised $\alpha_j\in\Fqq(\alpha_j)\cong\Fqr{2r_j}$ and $\lambda_j\in\Fqr{2}(\alpha_j^q+\alpha_j^{-q})\cong\Fqr{r_j}$.

\medskip

iv) For $j\in J_{2}$ we have that $\overline{f_j}=f_j\neq  f_j^{*}$ and Theorem~\ref{dih_Euclidean_dual}~(iv) applies in this case.  For $X\in \{\pmb{0}, I\}$ we have that $ \langle X\rangle^{\perp_E}= \langle I-X\rangle$ and consequently we take $\widetilde{X}=I-X$.

If $X=\left[\begin{smallmatrix} 0&1\\ 0&0\end{smallmatrix}\right]$, then $\langle X\rangle = \langle X\rangle^{\perp_E}$ by Theorem~\ref{dih_Euclidean_dual}~(iv), so applying $\varphi_j=\mathcal{H}^{r_j}$ in Theorem~\ref{theo_commutative} yields $\langle X\rangle^{\perp_H} = \langle X\rangle$.

Finally, if $X= \left[\begin{smallmatrix} 1&\lambda_j\\ 0&0\end{smallmatrix}\right]$ with $\lambda_j\in \Fqq(\alpha_j)$, then $\langle X\rangle^{\perp_E}=\left\langle \left[\begin{smallmatrix} 1&-\lambda_j\\ 0&0\end{smallmatrix}\right]\right\rangle$, and the image by $\varphi_j=\mathcal{H}^{r_j}$ yields the desired conclusion.

\medskip

v) For $j\in J_{3}$ we have that $f_j\neq  f_j^{*}=\overline{f_j}$, so the Euclidean orthogonals are, as in case $j\in J_2$, computed with Theorem~\ref{dih_Euclidean_dual}~(iv).
 However, in this case $\varphi_j=\mathcal{T}^{r_j}$ (see Theorem~\ref{theo_commutative}), and so: $$\varphi_j\left(\left\langle \begin{bmatrix}0&1\\0&0\end{bmatrix}\right\rangle\right) = \left\langle \begin{bmatrix}0&0\\ 1&0\end{bmatrix}\right\rangle = \left\langle \begin{bmatrix}1&0\\0&0\end{bmatrix}\right\rangle$$ and $$\varphi_j\left(\left\langle \begin{bmatrix}1&-\lambda_j\\0&0\end{bmatrix}\right\rangle\right) = \left\langle \begin{bmatrix}0&0\\ -\lambda_j^{q^{r_j}}&1\end{bmatrix}\right\rangle = \left\langle \begin{bmatrix}-\lambda_j^{q^{r_j}}&1\\0&0\end{bmatrix}\right\rangle,$$ so the claim is clear.

\medskip

vi) For $j\in J_{4}$,  since $f_j, \overline{f_j}$ and $f_j^{*}$ are pairwise different, we can argue as in iii) for $j\in J_1$, but taking into account that the Euclidean orthogonals correspond now to polynomials that are not self-reciprocal, so Theorem~\ref{dih_Euclidean_dual}~(iv) applies.
\end{proof}


\subsection{Computing Hermitian self-orthogonal \texorpdfstring{$D_n$}{Dn}-codes}

The algebraic description of the Hermitian dual of a $D_n$-code obtained in Theorem \ref{dih_hermitic_orthogonal}  is useful in several code constructions. For example, it can be used to derive, for every $D_n$-code over $\Fqq$, the explicit representation of its Hermitian hull (\emph{i.e.} the intersection with its Hermitian dual) and so the number of all distinct Hermitian Linear Complementary Dual $D_n$-codes over $\Fqq$ (\emph{cf.} \cite[Theorems 4.1 and 4.2]{CaoCaoFu23}). 

Another application of Theorem \ref{dih_hermitic_orthogonal}, which is fundamental in this paper, is that it allows one to detect $D_n$-codes that are Hermitian self-orthogonal. Recall that Hermitian self-orthogonal codes have a significant role in the construction of quantum CSS codes (see subsection~\ref{pre:CSS}). With this objective in mind, we proceed to characterise, through the Wedderburn-Artin decomposition of $\Fqq[D_n]$ with $\car(\Fqq)\nmid n$ (see Theorem \ref{theo_dih_refined}), the structure of Hermitian self-orthogonal $D_n$-codes.

\begin{theorem}
\label{dih-self-hermitic}
Assume \emph{Hypotheses~\ref{hipotesis_descomposicion_Fq2Dn}}. Following the notation in \emph{Theorem~\ref{dih_hermitic_orthogonal}}, consider the expression $I_{\cC}\  \overset{\rho}{\cong}\bigoplus_{j\in J} I_j\subseteq \bigoplus_{j\in J} A_j$ of an ideal $I_{\cC}\subseteq\Fqq[D_n]$ associated with a $D_n$-code $\mathcal{C}$ over $\Fqq$. Then $\mathcal{C}\subseteq \mathcal{C}^{\perp_H}$ if and only if the following conditions are satisfied:
\begin{enumerate}[label=\emph{\roman*)}]
\setlength{\itemsep}{1mm}
\item If $j\in J_0$ and $\car(\Fqq)\neq 2$, then $I_j= \pmb{0}\oplus \pmb{0}$.
\item If $j\in J_0$ and $\car(\Fqq)= 2$, then $I_j= \pmb{0}$ or $I_j=\langle \left[ \begin{smallmatrix} 1&1\\1&1\end{smallmatrix}\right]\rangle_{\Fqq[C_{2}]}$.
\item If $j\in J_1$, then either one of the summands of $I_j$ is $\pmb{0}$, or one of the following options holds for $I_j$:
	\begin{itemize}
		\item[$\blacktriangleright$] $\left\langle\begin{bmatrix} 0&1\\0&0 \end{bmatrix}\right\rangle \oplus \left\langle \begin{bmatrix} 2&-(\alpha_j^q+\alpha_j^{-q})\\0&0\end{bmatrix}\right\rangle$.
		\item[$\blacktriangleright$] $\left\langle\begin{bmatrix} 1&\lambda_j\\0&0 \end{bmatrix}\right\rangle \oplus \left\langle \begin{bmatrix} (2\lambda_j+\alpha_j+\alpha_j^{-1})^q & (-2-\lambda_j(\alpha_j+\alpha_j^{-1}))^q)\\0&0\end{bmatrix}\right\rangle$, $\lambda_j\in \Fqr{2}(\alpha_j+\alpha_j^{-1})$.
	\end{itemize}
\item If $j\in J_2$, then $I_j$ is either $\pmb{0}$, $\left\langle\left[ \begin{smallmatrix} 0&1\\0&0 \end{smallmatrix}\right]\right\rangle$, or $\left\langle\left[ \begin{smallmatrix} 1&\lambda_j\\0&0 \end{smallmatrix}\right]\right\rangle$ with $\lambda_j=-\lambda_j^{q^{r_j}}$, $\lambda_j\in\Fqr{2}(\alpha_j)$.
\item If $j\in J_3$, then $I_j$ is either $\pmb{0}$, or $\left\langle\left[ \begin{smallmatrix} 1&\lambda_j\\0&0 \end{smallmatrix}\right]\right\rangle$ with $0\neq \lambda_j=-\lambda_j^{-q^{r_j}}$, $\lambda_j\in\Fqr{2}(\alpha_j)$.
\item If $j\in J_4$, then either one of the summands of $I_j$ is $\pmb{0}$, or one of the following options holds for $I_j$:
	\begin{itemize}
		\item[$\blacktriangleright$] $\left\langle\begin{bmatrix} 0&1\\0&0 \end{bmatrix}\right\rangle \oplus \left\langle \begin{bmatrix} 0&1\\0&0\end{bmatrix}\right\rangle$.
		\item[$\blacktriangleright$] $\left\langle\begin{bmatrix} 1&\lambda_j\\0&0 \end{bmatrix}\right\rangle \oplus \left\langle \begin{bmatrix} 1&-\lambda_j^q\\0&0\end{bmatrix}\right\rangle$, $\lambda_j\in \Fqr{2}(\alpha_j)$.
	\end{itemize}
\end{enumerate}
\end{theorem}

\begin{proof}
In virtue of the isomorphism $\rho$ provided in Theorem \ref{theo_dih_refined} and  the formula $(\ref{reduccioDual})$, one has that $\mathcal{C}\subseteq \mathcal{C}^{\perp_H}$ if and only if ${\cal I}_j\subseteq {\cal I}_j^{\perp_H}$, for all $j\in J$. Therefore, we will now analyse  when this feature occurs, for each $j\in J$. 

For $j\in J_0$, the assertion is trivial in view of Theorem~\ref{dih_hermitic_orthogonal}. For $j\in J_1$, if $I_j=\langle X\rangle \oplus \langle Y\rangle$ then $I_j^{\perp_H}=\langle \widetilde{Y}\rangle \oplus \langle \widetilde{X} \rangle$ (see Theorem \ref{dih_hermitic_orthogonal})  and therefore $I_j\subseteq {I}_j^{\perp_H}$ if and only if 
$$  
\langle X \rangle \subseteq \langle \widetilde{Y} \rangle = (\sigma_j\circ \mathcal{H}^{2r_j-1}\circ \overline{\sigma_j}^{\, -1}) (\langle Y\rangle^{\perp_E})
$$ 
and 
$$  
\langle Y \rangle \subseteq \langle \widetilde{X} \rangle = (\overline{\sigma_j}\circ \mathcal{H}^{1}\circ \sigma_j^{-1}) (\langle X\rangle^{\perp_E}).
$$ 
Thus, if one of the summands of $I_j$ is $\pmb{0}$, as its Euclidean dual would be the whole matrix ring, then both inclusions are certainly satisfied and we are done. So we may suppose that none of the summands of $I_j$  is $\pmb{0}$. Observe that neither $X$ nor $Y$ are the identity matrix, since otherwise either $\langle \widetilde{X}\rangle$ or $\langle \widetilde{Y}\rangle$ would be $\pmb{0}$, which would force that one of the summands of $I_j$ is $\pmb{0}$, a contradiction. It follows that $X$ and $Y$ are matrices of rank 1, and therefore it must be verified: 
$$ 
 \langle X \rangle = \langle \widetilde{Y}\rangle =  (\sigma_j\circ \mathcal{H}^{2r_j-1}\circ \overline{\sigma_j}^{\, -1} )(\langle Y\rangle^{\perp_E})
 $$
 and 
 $$
 \langle Y \rangle = \langle \widetilde{X}\rangle = (\overline{\sigma_j}\circ \mathcal{H}^{1}\circ \sigma_j^{-1}) (\langle X\rangle^{\perp_E}).
 $$ 
 It is easy to see that these equalities are equivalent, since $$\left(\sigma_j\circ \mathcal{H}^{2r_j-1}\circ \overline{\sigma_j}^{\, -1} \right) \circ \left(\overline{\sigma_j}\circ \mathcal{H}^{1}\circ \sigma_j^{-1}\right) = \operatorname{Id}.$$ Consequently $I_j= \langle X \rangle \oplus \langle \widetilde{X}\rangle$, and the possible combinations follow from Theorem~\ref{dih_hermitic_orthogonal}.

For $j\in J_2$,  in view of  Theorem~\ref{dih_hermitic_orthogonal}, we have that  ${\cal I}_j\subseteq {\cal I}_j^{\perp_H}$ if and only if one of the situations claimed in the statement  occurs. Take now $I_j$ for $j\in J_3$. Then $I_j\subseteq {I}_j^{\perp_H}$ if and only if $I_j=\pmb{0}$, or 
 $$
 \left\langle \begin{bmatrix}1&\lambda_j\\0&0\end{bmatrix}\right\rangle = I_j = {I}_j^{\perp_H}= \left \langle\begin{bmatrix}-\lambda_j^{q^{r_j}}&1\\0&0\end{bmatrix}\right\rangle = \left\langle \begin{bmatrix}1&-\lambda_j^{-q^{r_j}}\\0&0\end{bmatrix}\right\rangle,
 $$ 
 where the second equality only holds for $\lambda_j\neq 0$, as wanted. Finally, the case $j\in J_4$ can be reasoned similarly as the situation $j\in J_1$.
\end{proof}

\medskip

As an immediate consequence, we can count the number of Hermitian self-orthogonal $D_n$-codes over $\Fqq$. In order to simplify the computations, we will use the next result.

\begin{proposition}
\label{prop_solutions}
Let $x\in \Fqr{2r}$ with $r\in \mathbb{N}$, and let $\xi\in \Fqr{2r}$ be a primitive element.
\begin{enumerate}
\setlength{\itemsep}{0mm}
\item[\emph{(a)}] If $q$ is odd, then $x=-x^{q^r}$ if and only if $x=0$ or $x=\xi^{\left(k+\frac{1}{2}\right)\left(q^r+1\right)}$ for $k\in \{0,1,...,q^r-2\}$.
\item[\emph{(b)}] If $q$ is odd, then $x=-x^{-q^r}$ if and only if $x=\xi^{\left(k+\frac{1}{2}\right)\left(q^r-1\right)}$ for $k\in \{0,1,...,q^r\}$.
\item[\emph{(c)}] If $q$ is a power of $2$, then $x=x^{-q^r}$ if and only if $x=\xi^{k\left(q^r-1\right)}$ for $k\in \{0,1,...,q^r\}$.
\end{enumerate}
\end{proposition}

\begin{proof}
We only demonstrate (a), since (b) and (c) are analogous. We may assume that $x\neq 0$, and $x=\xi^t$ for certain integer $t\in\{0,1\dots, q^{2r}-2\}$. Hence we should analyse when it holds $\xi^{t\left(q^r-1\right)}=-1$. But $\xi$ has order $q^{2r}-1$, so we can write $-1=\xi^{\frac{1}{2}\left(q^{2r}-1\right)}$. If we compare both equalities, we obtain that $t\left(q^r-1\right)=\frac{1}{2}\left(q^{2r}-1\right)$ modulo $q^{2r}-1$, that is, 
$$
t\left( q^r-1\right) = \tfrac{1}{2}\left(q^{2r}-1\right) + k \left(q^{2r}-1\right) \; \Longleftrightarrow \; t=\left(k+\tfrac{1}{2}\right)\left( q^r+1\right),
$$ for some integer $k$. Finally, it is not difficult to see that then $t\in\{0,1\dots, q^{2r}-2\}$ if and only if $k\in \{0,1,...,q^r-2\}$.
\end{proof}

\medskip

\begin{corollary}
\label{cor_dih_sel}
Assume \emph{Hypotheses~\ref{hipotesis_descomposicion_Fq2Dn}}. The number of Hermitian self-orthogonal $D_n$-codes over $\Fqq$ is 
$$ \epsilon\displaystyle\prod_{j\in J_1} \left(3q^{r_j}+6\right)  \displaystyle\prod_{j\in J_2\cup J_3} \left(q^{r_j}+2\right)    \displaystyle\prod_{j\in J_4} \left(3q^{2r_j}+6\right) ,
$$
with $\epsilon=1$ if $q$ is odd and $\epsilon=2$ otherwise.
\end{corollary}

\begin{proof}
It is enough to count how many ideals we can take in Theorem~\ref{dih-self-hermitic}, for each $j\in J$:
\begin{itemize}
	\item For $j\in J_0$ there is only one possibility if $q$ is odd, and two possibilities if $q$ is a power of $2$.
	\item For $j\in J_1$ we have that $A_j$ is isomorphic to $\mathcal{M}_2\left(\Fqr{r_j}\right) \oplus \mathcal{M}_2\left(\Fqr{r_j}\right)$. By Theorem~\ref{dih-self-hermitic} iii), we must first compute  the number of ideals $I_j$ in $A_j$ such that one summand is $\pmb{0}$. Since  there are $q^{r_j}+3$ possible ideals in $\mathcal{M}_2\left(\Fqr{r_j}\right)$ and  the ideal $\pmb{0}\oplus\pmb{0}$ cannot be computed twice, we obtain $2\left(q^{r_j}+3\right)-1$ possibilities for $I_j$ of this type. Apart from these cases, in view of Theorem~\ref{dih-self-hermitic} iii), we see that there are $q^{r_j}+1$ additional possibilities for $I_j$. To sum up, for $j\in J_1$, there are $3q^{r_j}+6$ possibilities.
	\item For $j\in J_2$, taking into account Theorem~\ref{dih-self-hermitic} iv) we see that, apart from $\pmb{0}$ and $\left\langle\left[ \begin{smallmatrix} 0&1\\0&0 \end{smallmatrix}\right]\right\rangle$, we must distinguish  the parity of $q$ to compute the other possibilities for $I_j$. If $q$ is odd, then Proposition~\ref{prop_solutions} (a)  yields $q^{r_j}-1$ additional ideals  $I_j$; in contrast, when $q$ is a power of $2$, the number of solutions of $\lambda_j=\lambda_j^{q^{r_j}}$ in $\Fqr{2r_j}$ is just the cardinality of $\Fqr{r_j}$, and so there are also $q^{r_j}-1$ additional ideals in this situation. To sum up, there are $q^{r_j}+2$ possibilities for $I_j$ with $j\in J_2$.
	\item For $j\in J_3$ we can reason similarly using Proposition~\ref{prop_solutions} (b-c), and we obtain $q^{r_j}+2$ possibilities in both cases when $q$ is odd or a power or $2$.
	\item For $j\in J_4$, we can reason as for $j\in J_1$, just noting that in this case the matrix rings take values in $\Fqr{2r_j}$.
\end{itemize}

\end{proof}

\smallskip

\begin{remark}
\label{remark_mistake}
In \cite{CaoCaoFu23} it is also computed the number of Hermitian self-orthogonal $D_n$-codes over $\Fqq$. Unfortunately, in view of the previous corollary, the computation of that number in \cite[Theorem 4.3]{CaoCaoFu23} does not seem to be correct.  In fact, we illustrate the inaccuracy of their formula in Example \ref{example-caocao}.
\end{remark}


\section{On Euclidean dualities of generalised quaternion codes}
\label{sec:quaternion}

Inspired by the above ideas, our goal in this section is to present the concrete representation of the Euclidean dual code of any $Q_n$-code, where $Q_{n}$ is the generalised quaternion group of order $4n$ ($n\geq 1$), in terms of the decomposition of the corresponding ideal in the group algebra $\Fq[Q_n]$ in the semisimple case, that is, when $\car(\Fq)$ does not divide $4n$. Up to our knowledge, this has not been addressed previously anywhere. Regarding the Hermitian dualities, as we will explain below (see Theorem \ref{thm:alg_isom} and Remark \ref{rem:q2igual}), the algebras $\Fqq[Q_n]$ and $\Fqq[D_{2n}]$ are isomorphic when $\car(\Fq)\nmid 4n$. Therefore, the computation of the Hermitian dual of a $Q_n$-code can be done using the information provided in section 3. Our main result here will be Theorem~\ref{quat_euclid_orthogonal}.

Throughout this section $\car(\Fq)\nmid 4n$, so in particular $q$ is always an odd number.

\subsection{The Wedderburn-Artin decomposition of  \texorpdfstring{$\F_{q}[Q_n]$}{Fq[Qn]}}

The purpose of this subsection is to present the Wedderburn-Artin decomposition of the semisimple group algebra $\Fq[Q_n]$ provided in \cite[Theorems 3.1 and 3.6]{GaoYue21}. However, before doing that,  it is worth noting the existence of the following criterion to decide when the semisimple group algebras $\Fq[Q_{n}]$ and $\Fq[D_{2n}]$ are isomorphic or not:

\begin{theorem}\emph{(\cite[Theorem 5.1]{Flaviana2009})}
\label{thm:alg_isom}
If $\car(\Fq)$ does not divide $4n$, then the algebras $\Fq[Q_{n}]$ and $\Fq[D_{2n}]$ are isomorphic if and only if $n$ is even or $q\equiv 1 \text{\emph{ (mod $4$)}}$.
\end{theorem}

The authors of \cite{GaoYue21} were apparently not aware of that criterion, and they computed in \cite[Theorem 3.1]{GaoYue21} the structure of $\Fq[Q_n]$ when $q\equiv 1 \text{ (mod $4$)}$, which actually was already known by \cite[Theorem 3.1]{Bro15} (see (\ref{eq:isom_decom_diedric})). 

Consequently, since we are particularly assuming that $q$ is odd, then by Theorem~\ref{thm:alg_isom} we may certainly restrict ourselves hereafter to the situation $q\equiv 3 \text{ (mod $4$)}$ and $n$ odd.

\begin{remark}\label{rem:q2igual}
We highlight that the Hermitian dual code of a $Q_n$-code over $\Fqq$, when $\car(\Fq) \nmid 4n$, is isomorphic to the Hermitian dual code of a $D_{2n}$-code over $\Fqq$ due to Theorem~\ref{thm:alg_isom}, since in both cases when $q$ is congruent with $1$ or $3$ modulo $4$ it holds $q^2\equiv 1 \text{ (mod 4)}$. Therefore, we may apply Theorem~\ref{dih_hermitic_orthogonal} for computing the Hermitian dual of a $Q_n$-code over $\Fqq$, and so along this section we focus in the Euclidean duality of $Q_n$-codes. 
\end{remark}

Prior to presenting the isomorphism provided in \cite[Theorem 3.6]{GaoYue21} for $q\equiv 3 \text{ (mod $4$)}$, we first establish the notation for rest of this section, and we prove some preliminary lemmas.

\smallskip

\begin{hypotheses} 
\label{hipotesis_descomposicion_FqQn}
Suppose that $n$ is odd, that $\car(\Fq)\nmid 4n$ and that $q\equiv 3 \text{\emph{ (mod $4$)}}$. Decompose the polynomials $\mx^n-1$ and $\mx^n+1$ into irreducible monic factors in $\Fq[\mx]$ as follows:
\begin{eqnarray*}
\mx^n-1 &=& f_1f_2 \cdots f_r f_{r+1}f_{r+1}^* \cdots f_{r+s}f_{r+s}^*, \\
\mx^n+1&=& g_1g_2 \cdots g_t g_{t+1}g_{t+1}^* \cdots g_{t+k}g_{t+k}^*,
\end{eqnarray*}
where $f_1 := \mx-1$, $g_1 := \mx+1$, $f_i=f_{i}^*$ and $g_j=g_j^*$ for $1\leq i\leq r$ and $1\leq j \leq t$, respectively. Further, denote by $\alpha_i$ any root of $f_i$, and by $\beta_j$ any root of $g_j$. 
\end{hypotheses}

\smallskip

\begin{lemma}
\label{raiz-1}
If $\car(\Fq)\neq 2$, then $\sqrt{-1}\in \Fq$ if and only if $q\equiv 1 \text{ \emph{(mod $4$)}}$. In particular, if $q\equiv 3 \text{\emph{ (mod $4$)}}$, then $\sqrt{-1}\in\Fqr{2}\smallsetminus \Fq$.
\end{lemma}

\begin{proof}
Note that $\sqrt{-1}\in \Fq$ implies that the multiplicative group $\Fq^{\times}$ contains an element of order $4$, so $4$ divides $q-1$. Reciprocally, if $q-1=4k$ for some positive integer $k$, then $(\sqrt{-1})^{q-1}=(\sqrt{-1})^{4k}=1$, so $\sqrt{-1}\in \Fq$. Finally, observe that $q\equiv 3 \text{ (mod $4$)}$ leads to $q^2\equiv 1 \text{ (mod $4$)}$, and therefore $\sqrt{-1}\in \Fqr{2}\smallsetminus \Fq$ by the first assertion.
\end{proof}

Note that, under the assumptions in Hypotheses~\ref{hipotesis_descomposicion_FqQn} the above result tells us in particular that  
\begin{equation}\label{eq:-1}
\sqrt{-1}\in \Fqq\leqslant  \Fq(\beta_j), \; \mbox{ for any } 2\leq j\leq t,
\end{equation}
since in this case $g_{j}$ is a self-reciprocal polynomal with even degree (see Remark \ref{arrels_recipr}).

\begin{lemma}
\label{lemma:gamma}
Assume \emph{Hypotheses~\ref{hipotesis_descomposicion_FqQn}}. The matrices $$Z_i:=\left[\begin{smallmatrix} 1 & -\alpha_j \\ 1&-\alpha_j^{-1}\end{smallmatrix}\right]\in {\cal M}_2\left(\Fq(\alpha_i)\right) \mbox{ and } \; Y_j:=\left[\begin{smallmatrix} \sqrt{-1} & -\beta_j \\ \sqrt{-1} & -\beta_j^{-1} \end{smallmatrix}\right]\in{\cal M}_2\left(\Fq(\beta_j)\right),$$ for $2\leq i \leq r$ and $2\leq j\leq t$, respectively, are invertible. Moreover, the maps given by $$\sigma_i(X) := Z_i^{-1}X Z_i \; \mbox{ and } \;\gamma_j(X) := Y_j^{-1} X Y_j$$ are automorphisms of ${\cal M}_2\left(\Fq(\alpha_i)\right)$ and ${\cal M}_2\left(\Fq(\beta_j)\right)$, respectively.
\end{lemma}

\begin{proof}
In Lemma~\ref{lemma:sigmas} it was already proved that $Z_i$ is invertible. Moreover, it is clear from $(\ref{eq:-1})$ that $Y_j\in \mathcal{M}_2(\Fq(\beta_j))$, and its determinant is $\sqrt{-1}(\beta_j-\beta_j^{-1})$, which cannot be zero: otherwise $\beta_j=\beta_j^{-1}$ so $\beta_j^2-1=0$, which is a contradiction since $g_j$ is not a divisor of $\mx^2-1$. Finally, the maps $\sigma_i$ and $\gamma_j$ are certainly automorphisms because they are conjugation maps. 
\end{proof}

\begin{remark}
\label{remark_divide}
As we will see in Theorem~\ref{theo_quat_GaoYue}, which recalls the Wedderburn-Artin decomposition of $\F_q[Q_n]$ given in \cite{GaoYue21}, the isomorphisms corresponding to the reciprocal polynomials $g_j$ for $2\leq j\leq t$ will involve a conjugation by the previous map $\gamma_j$, and its image will be $\mathcal{M}_2(\Fq(\beta_j+\beta_j^{-1}))$; nevertheless, $\sqrt{-1}\notin\Fq(\beta_j+\beta_j^{-1})\cong\Fqr{\deg(g_j)/2}$ if $4$ does not divide $\deg(g_j)$, due to Lemma~\ref{raiz-1}. Thus, in order to be well-defined in that specific case, another conjugation map is needed, which is introduced in Lemma~\ref{lemma:theta} below. The next well-known result is useful for that aim. We include its proof for the sake of comprehensiveness. 
\end{remark}

\begin{lemma}
\label{solutions_FF}
In any finite field $\Fq$, the equation $u^2+v^2=k$ has solutions, for every $k\in \Fq$.
\end{lemma}

\begin{proof}
If $\car(\Fq)=2$, then the multiplicative group $\Fq^{\times}$ has odd order. Hence the homomorphism $f:\Fq^{\times}\longrightarrow\Fq^{\times}$ given by $f(x)=x^2$ has trivial kernel, which implies that $f$ is an automorphism and thus every element in $\Fq^{\times}$ is a square. In particular, taking $u=0$, then $v^2=k$ has solutions for every $k\in \Fq$.

Suppose now that $\car(\Fq)\neq 2$. In this case the above defined homomorphism $f$ has kernel of order $2$ (because $\Fq^{\times}$ is cyclic), so the number of squares in $\Fq^{\times}$ is $\tfrac{q-1}{2}$. Therefore, joint with $0$, we get that the number of squares in $\Fq$ is $\tfrac{q+1}{2}$. Fix $k\in \Fq$ and set $S:=\{u^2\,|\, u\in \Fq\}$ and $T:=\{k-u^2\,|\, u\in \Fq\}$. Note that both sets have size $\frac{q+1}{2}$, and since they cannot be disjoint, they must coincide in some element, that is, there exist $u,v\in\Fq$ such that  $u^2=k-v^2$, as wanted.
\end{proof}

\begin{lemma}
\label{lemma:theta}
Assume \emph{Hypotheses~\ref{hipotesis_descomposicion_FqQn}}. If $2 \leq j \leq t$ and $4 \nmid \deg(g_j)$, then the next statements hold:
\begin{enumerate}[label=\emph{(\alph*)}]
\item The set 
 $$S_j:=\left\{ \begin{bmatrix}
w & z \\ -z^{q^{\deg(g_j)/2}} & w^{q^{\deg(g_j)/2}}
\end{bmatrix} \; \middle\vert \; w,z \in \Fq(\beta_j)\right\}$$
is a subring of $\mathcal{M}_2(\Fq(\beta_j))$ and it satisfies that 
\[
S_j=\left\{  \begin{bmatrix}
x_1+x_2\sqrt{-1} & x_3+x_4\sqrt{-1} \\ -(x_3-x_4\sqrt{-1}) & x_1-x_2\sqrt{-1}
\end{bmatrix} \; \middle\vert \;   x_1, x_2, x_3, x_4\in \Fq(\beta_j+\beta_j^{-1})  \right\}.
\]

\item Fix  $u, v \in \Fq(\beta_j+\beta_j^{-1})$ such that $u^2+v^2 = -1$. The map $\theta_j: S_j\longrightarrow \mathcal{M}_2\left(\Fq(\beta_j+\beta_j^{-1})\right)$ such that 
\[
\theta_j:
\begin{bmatrix}
x_1+x_2\sqrt{-1} & x_3+x_4\sqrt{-1} \\ -(x_3-x_4\sqrt{-1}) & x_1-x_2\sqrt{-1}
\end{bmatrix}
\mapsto
\begin{bmatrix}
x_1+x_2u-x_4v & x_3+x_2v+x_4u \\ -x_3+x_2v+x_4u & x_1-x_2u+x_4v
\end{bmatrix}
\]
is a ring isomorphism. 
\end{enumerate}
\end{lemma}

\begin{proof}
It is straightforward to prove that $S_{j}$ is a ring. To conclude statement (a) we show below that any matrix in $S_j$ can be written as 
$$\begin{bmatrix}
x_1+x_2\sqrt{-1} & x_3+x_4\sqrt{-1} \\ -(x_3-x_4\sqrt{-1}) & x_1-x_2\sqrt{-1}
\end{bmatrix}$$ 
for certain $x_1, x_2, x_3, x_4\in \Fq(\beta_j+\beta_j^{-1})$. 

First, since $\deg(g_j)$ is an even number not divisible by $4$ and $\Fq(\beta_j+\beta_j^{-1})\cong \Fqr{\deg(g_j)/2}$, then $\sqrt{-1}\in \Fq(\beta_j)\smallsetminus \Fq(\beta_j+\beta_j^{-1})$ by Lemma~\ref{raiz-1}. So $\Fq(\beta_j)/ \Fq(\beta_j+\beta_j^{-1})$ is a field extension of degree $2$ with basis $\{1, \sqrt{-1}\}$, and it follows that any element of $\Fq(\beta_j)$ can be written as $x+y\sqrt{-1}$ for certain $x,y\in \Fq(\beta_j+\beta_j^{-1})$. 

Furthermore, the fact $\sqrt{-1}\notin \Fq(\beta_j+\beta_j^{-1})$ implies that $(\sqrt{-1})^{q^{\deg(g_j)/2}}\neq \sqrt{-1}$. In addition, since $q^{\deg(g_j)/2}$ is an odd prime power,  $(\sqrt{-1})^{q^{\deg(g_j)/2}}$ should be the unique other element of order $4$ different from $\sqrt{-1}$ in the cyclic group $\Fqr{\deg(g_j)}^{\times}$, so it should be equal to $-\sqrt{-1}$. Consequently,
$$
(x+y\sqrt{-1})^{q^{\deg(g_j)/2}}=x+y(\sqrt{-1})^{q^{\deg(g_j)/2}}=x-y\sqrt{-1},
$$ 
for any $x,y\in \Fq(\beta_j+\beta_j^{-1})$.  Now, the double expression of the elements of $S_{j}$ stated in $(a)$ is clear, as claimed.

Statement (b) is demonstrated in \cite[Lemma 3.5]{GaoYue21};  we only highlight that the existence of $u,v\in \Fq(\beta_j+\beta_j^{-1})$ is due to Lemma~\ref{solutions_FF}.
\end{proof}

\begin{lemma}
\label{tech-power}
Assume \emph{Hypotheses~\ref{hipotesis_descomposicion_FqQn}}. For any $2\leq j\leq t$ and any root $\beta_j$ of $g_j$ it holds $\beta_j^{-1}=\beta_j^{q^{\deg(g_j)/2}}$. 
\end{lemma}

\begin{proof}
For $2\leq j\leq t$ the polynomial $g_j$ is self-reciprocal. Then it has even degree and $\beta_j^{-1}$ is also a root of $g_j$ (see Remark \ref{arrels_recipr}). Thus, there exists $k\in \{0,\dots, \deg(g_j)-1\}$ such that $\beta_j^{-1}=\beta_j^{q^k}$. It follows that $\beta_j=\beta_j^{q^{2k}}$ and therefore $\beta_j\in \Fqr{2k}$. This implies that $\deg(g_j)$ must divide $2k$ and we obtain that necessarily $2k=\deg(g_j)$.
\end{proof}

We are now ready to present the Wedderburn-Artin decomposition of the semisimple algebra $\Fq[Q_n]$ provided in \cite[Theorem 3.6]{GaoYue21} when $q\equiv 3 \text{ (mod 4)}$, with a misprint corrected (concretely, in $B_1$ below). It should be noted that this decomposition is closely related to the decomposition of the dihedral group algebra $\Fq[D_n]$, (see  (\ref{eq:isom_decom_diedric})). This is why, with some technical variations,  the isomorphism $\rho$ of (\ref{eq:isom_decom_diedric}) (see also Theorem \ref{theo_dih_refined}) forms the ``dihedral part'' of the following algebra isomorphism $\psi$, corresponding to the decomposition of $\Fq[Q_n]$.

\begin{theorem}
\label{theo_quat_GaoYue}
Assume \emph{Hypotheses~\ref{hipotesis_descomposicion_FqQn}}. There exists an isomorphism of $\Fq$-algebras \begin{equation}\label{eq:isom_quaterni}
\psi:\Fq[Q_n] \longrightarrow \bigoplus_{i=1}^{r+s} A_i \oplus \bigoplus_{j=1}^{t+k} B_j,
\end{equation}
where every $A_i$ is given as in the dihedral case (see \emph{(\ref{eq:decom_diedric})}) and 
\begin{equation}\label{eq:decom_quaterni}
    B_j:= 
    \begin{cases}
        \Fq(\sqrt{-1}) & \text{if } j=1 \\[1mm]
        {\cal M}_2\left(\Fq(\beta_j+\beta_j^{-1}) \right) & \text{if } 2\leq j \leq t \\[2mm]
         {\cal M}_2\left(\Fq(\beta_j) \right) & \text{if } t+1 \leq j \leq t + k
    \end{cases} .
\end{equation}
Further, $\psi = \bigoplus_{i=1}^{r+s} \rho_i \oplus \bigoplus_{j=1}^{t+k} \eta_j$, where each $\rho_i$ and each $\eta_j$ is defined by the generators of $Q_{n} = \langle a, b \, | \, a^{2n} = 1, b^2 = a^n, b^{-1}ab = a^{-1} \rangle$ as follows: 
\begin{enumerate}[label=\emph{\roman*)}]
\item for $i=1$
$$\rho_i(a):=(1,1) \quad \text{ and } \rho_i(b):=(1,-1).$$

\item for $2\leq i \leq r$
$$\rho_i(a):=\sigma_i \left( \begin{bmatrix} \alpha_i & 0 \\ 0 & \alpha_i^{-1} \end{bmatrix} \right)  \quad \text{and} \quad  \rho_i(b):=\sigma_i \left( \begin{bmatrix} 0&1  \\  1 &0 \end{bmatrix} \right),$$where the map $\sigma_i$ is the one defined in \emph{Lemma~\ref{lemma:gamma}}. 

\item for $r+1\leq i \leq r+s$ 
$$\rho_i(a):= \begin{bmatrix} \alpha_i & 0 \\ 0 & \alpha_i^{-1} \end{bmatrix} \quad \text{and} \quad  \rho_i(b):= \begin{bmatrix} 0&1  \\  1 &0 \end{bmatrix}.$$

\item for $j=1$
$$
\eta_j(a):=-1 \quad \text{ and } \eta_j(b):=\sqrt{-1}.$$
 
\item for $2\leq j \leq t$ and $4\,\mid\, \deg(g_j)$
$$\eta_j(a):=\gamma_j \left( \begin{bmatrix} \beta_j & 0 \\ 0 & \beta_j^{-1} \end{bmatrix} \right)  \quad \text{and} \quad  \eta_j(b):=\gamma_j \left( \begin{bmatrix} 0&\sqrt{-1}  \\  \sqrt{-1} &0 \end{bmatrix} \right),$$where the map $\gamma_j$ is the one defined in \emph{Lemma~\ref{lemma:gamma}}. 

\item for $2\leq j \leq t$ and $4\,\nmid\, \deg(g_j)$
$$\eta_j(a):=\theta_j \left( \begin{bmatrix} \beta_j & 0 \\ 0 & \beta_j^{-1} \end{bmatrix} \right)  \quad \text{and} \quad  \eta_j(b):=\theta_j \left( \begin{bmatrix} 0&\sqrt{-1}  \\  \sqrt{-1} &0 \end{bmatrix} \right),$$where the map $\theta_j$ is the one defined in \emph{Lemma~\ref{lemma:theta}}. 

\item for $t+1\leq i \leq t+k$ 
$$\eta_j(a):= \begin{bmatrix} \beta_j & 0 \\ 0 & \beta_j^{-1} \end{bmatrix} \quad \text{and} \quad  \eta_j(b):= \begin{bmatrix} 0&-1  \\  1 &0 \end{bmatrix}.$$
\end{enumerate}
\end{theorem}

\begin{proof}
The maps $\rho_i$ are proved to be well-defined and algebra isomorphisms in \cite[Theorem 3.1]{Bro15}. Regar\-ding the maps $\eta_j$, they are proved to be well-defined and algebra isomorphisms in \cite[Theorems 3.1, 3.6]{GaoYue21}. However,  we point out a minor misprint within the proof of the injectivity of $\eta_1$ in \cite[Theorem 3.1]{GaoYue21}, concretely in the second direct summand, where it should appear $-\sqrt{-1}$ instead of $\sqrt{-1}$. Moreover, for the sake of comprehensiveness, we provide additional details to those in \cite{GaoYue21} concerning the fact that $\eta_j$ is well-defined for $2\leq j \leq t$, \emph{i.e.} that its image effectively belongs to  $\mathcal{M}_2 \left( \Fq(\beta_j + \beta_j^{-1}) \right)$.

On the one hand, if $4$ divides $\deg(g_j)$, then  $\sqrt{-1}\in\Fqq\leqslant\Fq(\beta_j+\beta_j^{-1})\cong \Fqr{\deg(g_j)/2}$ by Lemma~\ref{raiz-1}  (see also $(\ref{eq:-1})$). Hence by Lemma~\ref{lemma:gamma} we get:
\begin{eqnarray*}
\eta_j(a) &=& \begin{bmatrix}
\sqrt{-1} & -\beta_j \\ \sqrt{-1} & -\beta_j^{-1}
\end{bmatrix}^{-1} \begin{bmatrix}
		\beta_j & 0 \\ 0 & \beta_j^{-1}
		\end{bmatrix}   \begin{bmatrix}
\sqrt{-1} & -\beta_j \\ \sqrt{-1} & -\beta_j^{-1}
\end{bmatrix} \\ & = &
		\dfrac{1}{\sqrt{-1}(\beta_j-\beta_j^{-1})}\begin{bmatrix}
-\beta_j^{-1} & \beta_j \\ -\sqrt{-1} & \sqrt{-1}
\end{bmatrix}
		 \begin{bmatrix}
		\sqrt{-1}\beta_j & -\beta_j^2 \\ \sqrt{-1}\beta_j^{-1} & -\beta_j^{-2}
		\end{bmatrix} \\ & = & \begin{bmatrix}
		0 & (\sqrt{-1})^{-1} \\ (\sqrt{-1})^{-1} & \beta_j + \beta_j^{-1} 
		\end{bmatrix}\in \mathcal{M}_2 \left( \Fq(\beta_j + \beta_j^{-1}) \right),
\end{eqnarray*}
and similarly
\begin{eqnarray*}
\eta_j (b) &=& \begin{bmatrix}
\sqrt{-1} & -\beta_j \\ \sqrt{-1} & -\beta_j^{-1}
\end{bmatrix}^{-1} \begin{bmatrix}
		0&\sqrt{-1} \\ \sqrt{-1}&0
		\end{bmatrix}   \begin{bmatrix}
\sqrt{-1} & -\beta_j \\ \sqrt{-1} & -\beta_j^{-1}
\end{bmatrix} \\ & = &
		\dfrac{1}{\sqrt{-1}(\beta_j-\beta_j^{-1})}\begin{bmatrix}
-\beta_j^{-1} & \beta_j \\ -\sqrt{-1} & \sqrt{-1}
\end{bmatrix}
		 \begin{bmatrix}
		-1& -\beta_j^{-1}\sqrt{-1} \\ -1 & -\beta_j\sqrt{-1}
		\end{bmatrix} \\ & = & \begin{bmatrix}
		-(\sqrt{-1})^{-1} & -(\beta_j+\beta_j^{-1})\\ 0 & (\sqrt{-1})^{-1}
		\end{bmatrix}\in \mathcal{M}_2 \left( \Fq(\beta_j + \beta_j^{-1}) \right).
\end{eqnarray*}

On the other hand, if $4$ does not divide $\deg(g_j)$, then as $\beta_j^{-1}=\beta_j^{q^{\deg(g_j)/2}}$ by Lemma~\ref{tech-power} we deduce that $\left[\begin{smallmatrix} \beta_j & 0 \\ 0& \beta_j^{-1}\end{smallmatrix}\right]\in S_j$, where $S_j$ is the matrix ring in Lemma~\ref{lemma:theta}~(a). In particular, the image by $\theta_j$ of this last matrix lies in $\mathcal{M}_2(\Fq(\beta_j+\beta_j^{-1}))$ by the same lemma. Finally, as remarked within the proof of Lemma~\ref{lemma:theta} (a), it holds $-(\sqrt{-1})^{q^{\deg(g_j)/2}}=-(-\sqrt{-1})=\sqrt{-1}$, so $\left[\begin{smallmatrix} 0& \sqrt{-1} \\ \sqrt{-1}& 0\end{smallmatrix}\right]\in S_j$ and its image by $\theta_j$ also belongs to $\mathcal{M}_2(\Fq(\beta_j+\beta_j^{-1}))$.
\end{proof}

\subsection{On the Euclidean dual of a \texorpdfstring{$Q_n$-code}{Qn-code}}

In this subsection we will obtain a full description of the Euclidean dual of a $Q_n$-code $\cC$ over $\Fq$ based on the expression of $I_{\cC}$ as a direct sum of ideals within the Wedderburn-Artin decomposition of the semisimple algebra $\F_{q}[Q_n]$  provided by Theorem \ref{theo_quat_GaoYue}. Specifically, if 
$$
\psi(I_{\cC})=\bigoplus_{i=1}^{r+s} I_i\oplus \bigoplus_{j=1}^{t+k} L_j \subseteq 
\bigoplus_{i=1}^{r+s} A_i \oplus \bigoplus_{j=1}^{t+k} B_j,
$$
then, similarly to how we worked in section \ref{DualHerDn}, here we also use $(\ref{reduccioDual})$ to reduce the problem of computing $(I_{\cC})^{\perp_E}$ to compute $(I_i)^{\perp_E}$ and $(L_j)^{\perp_E}$ in their corresponding matrix rings, for all $i$ and $j$. 
Note that  by (\ref{eq:orthogonal_ideal_Euclidean}) in $\Fq[Q_n]$ the Euclidean dual of  $I_{\cC}$ can be computed as $I_{\cC^{\perp_E}}=\widehat{\operatorname{Ann}_r(I_{\cC})}$. Therefore, considering the algebra antiautomorphism $\operatorname{inv}: x\mapsto \hat{x}$ of $\Fq[Q_n]$, our goal is to find $\nu$ such that the next diagram is commutative.


\begin{equation}
\label{diagrama_cuaternio}
\begin{tikzcd}[row sep=large, column sep=large]
\Fq[Q_n] \arrow[d, "\psi"'] \arrow[r, "\operatorname{inv}"] & \Fq[Q_n] \arrow[d, "\psi"]     \\
\displaystyle\bigoplus_{i=1}^{r+s}  A_i \oplus \bigoplus_{j=1}^{t+k} B_j \arrow[r, "\nu"]  & \displaystyle\bigoplus_{i=1}^{r+s} A_i \oplus \bigoplus_{j=1}^{t+k} B_j 
\end{tikzcd}
\end{equation}
In this manner,  we obtain
\begin{eqnarray*}
\bigoplus_{i=1}^{r+s}(I_i)^{\perp_E} \oplus \bigoplus_{j=1}^{t+k} (L_j)^{\perp_E}&=&\psi (I_{\cC^{\perp_E}})= (\psi \circ \operatorname{inv} )(\operatorname{Ann}_r(I_{\cC})) = (\nu \circ \psi )(\operatorname{Ann}_r(I_{\cC}))\\ &=& \nu \left(\bigoplus_{i=1}^{r+s}\operatorname{Ann}_r(I_i) \ \oplus \ \bigoplus_{j=1}^{t+k} \operatorname{Ann}_r(L_j) \right)
\end{eqnarray*}
and so we will be able to compute each $(I_i)^{\perp_E}$ and $(L_j)^{\perp_E}$ in terms of the images through $\nu$ of $\operatorname{Ann}_r(I_i)$ and $\operatorname{Ann}_r(L_j)$, respectively.

In order to define such a $\nu$, we need the following map $\mathcal{K}$ for square matrices over any of the fields appearing in $A_i$, that is, $\Fq(\alpha_i)$ or $\Fq(\alpha_i + \alpha_i^{-1})$, and also the map $\cal S$, for square matrices over any of the fields appearing in $B_j$, that is, $\Fq(\beta_j)$ or $\Fq(\beta_j+\beta_j^{-1})$, for $2\leq i\leq r+s$ and $2\leq j\leq t+k$ respectively: 

\begin{equation}\label{eq:KiS}
\mathcal{K} : \begin{bmatrix}
a & b \\ c & d
\end{bmatrix} \mapsto \begin{bmatrix}
d & b \\ c & a
\end{bmatrix} ,
\hspace{1.5cm}
\mathcal{S} : \begin{bmatrix}
a & b \\ c & d
\end{bmatrix} \mapsto \begin{bmatrix}
d & -b \\ -c & a
\end{bmatrix}.
\end{equation}

\begin{lemma}
\label{antiautomorphisms}
The maps $\mathcal{K}$ and $\mathcal{S}$ in \emph{(\ref{eq:KiS})} are antiautomorphisms in their corresponding matrix rings.
\end{lemma}

\begin{proof}
Notice that $\mathcal{K}^2=\operatorname{id}=\mathcal{S}^2$, so they are one-to-one. Moreover, observe that 
\begin{gather*}
\mathcal{K} \left( \begin{bmatrix} a_1 & b_1 \\ c_1 & d_1 \end{bmatrix} \right) \mathcal{K} \left( \begin{bmatrix} a_2 & b_2 \\ c_2 & d_2 \end{bmatrix} \right)
= \begin{bmatrix} d_1 & b_1 \\ c_1 & a_1 \end{bmatrix}\begin{bmatrix} d_2 & b_2 \\ c_2 & a_2 \end{bmatrix}
= \begin{bmatrix} d_1d_2 + b_1c_2 & d_1b_2+b_1a_2 \\ c_1d_2 + a_1c_2 & c_1b_2+a_1a_2 \end{bmatrix} \\
= \mathcal{K} \left( \begin{bmatrix} c_1b_2+a_1a_2 & d_1b_2+b_1a_2 \\ c_1d_2 + a_1c_2 & d_1d_2 + b_1c_2 \end{bmatrix} \right) = 
\mathcal{K} \left( \begin{bmatrix} a_2 & b_2 \\ c_2 & d_2 \end{bmatrix} \begin{bmatrix} a_1 & b_1 \\ c_1 & d_1 \end{bmatrix} \right) ,
\end{gather*}
and
\begin{gather*}
\mathcal{S} \left( \begin{bmatrix} a_1 & b_1 \\ c_1 & d_1 \end{bmatrix} \right) \mathcal{S} \left( \begin{bmatrix} a_2 & b_2 \\ c_2 & d_2 \end{bmatrix} \right)
= \begin{bmatrix} d_1 & -b_1 \\ -c_1 & a_1 \end{bmatrix}\begin{bmatrix} d_2 & -b_2 \\ -c_2 & a_2 \end{bmatrix}
\\ = \begin{bmatrix} d_1d_2 + b_1c_2 & -d_1b_2-b_1a_2 \\ -c_1d_2 - a_1c_2 & c_1b_2+a_1a_2 \end{bmatrix} 
= \mathcal{S} \left( \begin{bmatrix} c_1b_2+a_1a_2 & d_1b_2+b_1a_2 \\ c_1d_2 + a_1c_2 & d_1d_2 + b_1c_2 \end{bmatrix} \right) \\ = 
\mathcal{S} \left( \begin{bmatrix} a_2 & b_2 \\ c_2 & d_2 \end{bmatrix} \begin{bmatrix} a_1 & b_1 \\ c_1 & d_1 \end{bmatrix} \right) ,
\end{gather*}
as desired.
\end{proof}

We are now in a position to give the expression of the isomorphism $\nu$ that makes the diagram $(\ref{diagrama_cuaternio})$ commutative. As explained in subsection \ref{Euclidean_dihedral},  in \cite{VedDeu21} a precise description of the Euclidean dual of a $D_n$-code over $\Fq$ is given (see Theorem \ref{dih_Euclidean_dual}). To this end, the authors used a simplified version, adapted to the Euclidean case, of the diagram given in (\ref{diagrama}) (see  \cite[Lemma 1]{VedDeu21}). That is why  the ``dihedral part'' of our isomorphism $\nu$ will be denoted as $\varphi$ (see Theorem \ref{theo_commutative}) and its expression will be given without a proof. Furthermore, with regard to the ``second part'' of the isomorphism $\nu$,  the isomorphisms corresponding to the reciprocal polynomials $g_j$, for $2\leq j\leq t$ will involve the conjugation by the map $\gamma_j$ defined in Lemma \ref{lemma:gamma}, when $4$ divides de degree of $g_j$ and by the map $\theta_j$  defined in Lemma \ref{lemma:theta} in other case (see Remark \ref{remark_divide}).

\begin{theorem}
\label{theo_commutative_quaternion}
Assume \emph{Hypotheses~\ref{hipotesis_descomposicion_FqQn}}. The diagram \emph{(\ref{diagrama_cuaternio})} is commutative if $$\nu = \bigoplus_{i=1}^{r+s} \varphi_i \oplus \bigoplus_{j=1}^{t+k} \phi_j  $$

\begin{small}
$$\begin{cases}
\varphi_i(x \oplus y) := x \oplus y & \text{if } i= 1 \\

\varphi_i(X) := \left(\sigma_i\circ \mathcal{K} \circ \sigma_i^{-1}\right) (X) & \text{if } 2\leq i \leq r \\

\varphi_i(X) := \mathcal{K}(X) & \text{if } r+1\leq i \leq r+s \\

\phi_j(x) := x^q & \text{if } j=1 \\

\phi_j(X) := \left(\gamma_j \circ \mathcal{S}\circ \gamma_j^{-1}\right)(X) & \text{if } 2\leq j \leq t \text{ and } 4\,\mid\, \deg(g_j) \\

\phi_j(X) := \left(\theta_j \circ \mathcal{S}\circ \theta_j^{-1}\right)(X) & \text{if } 2\leq j \leq t \text{ and } 4\,\nmid\, \deg(g_j) \\

\phi_j(X) := \mathcal{S}(X) & \text{if } t+1\leq j \leq t+k
\end{cases},$$ 
\end{small}

\smallskip 

\noindent where $\sigma_i,\gamma_j$ and $\theta_j$ are the maps in \emph{Lemmas~\ref{lemma:gamma} and \ref{lemma:theta}}.
\end{theorem}

\begin{proof}
Consider the presentation $Q_{n} = \langle a, b \, | \, a^{2n} = 1, b^2 = a^n, b^{-1}ab = a^{-1} \rangle$  for the generalised quaternion group of order $4n$. Then 
$$
Q_n=\{1,a,a^2,...,a^{2n-1},b,ab,a^2b,...,a^{2n-1}b\},
$$ 
so any element $u\in \Fq[Q_n]$ can be written as 
$$
u=P_1(a) + P_2(a)a^n + \left( P_3(a) + P_4(a)a^n \right) b,
$$
for certain polynomials $P_1,P_2,P_3,P_4\in\Fq[\mx]$ of degree at most $n-1$. We aim to show that  $$
(\psi\circ\operatorname{inv})(u)=  \psi(\hat{u}) = \nu(\psi(u)),
$$
for all $u\in \Fq[Q_n]$. The ``dihedral part'' of this equality corresponds to $1\leq i \leq r+s$, that is $\rho_i(\hat{u})=\varphi_i(\rho_i(u))$ and it
is already done in \cite[Lemma 1]{VedDeu21}.  Therefore, we only focus on $1\leq j\leq t+k$, and so we aim to show $\eta_j(\hat{u})=\phi_j(\eta_j(u))$ for these values of $j$. Let us argue case by case. Note firstly that  in $Q_n$ one has $(a^mb)^{-1}=(ba^{-m})^{-1}=a^mb^3=a^{n+m}b$, which leads to
$$
\hat{u} = P_1(a^{-1}) + P_2(a^{-1})a^n + \left( P_4(a) + P_3(a)a^n \right) b.
$$

\medskip

\noindent \underline{$j=1$.}
 \begin{align*}
\eta_1(\hat{u}) &= \eta_1\left(P_1(a^{-1}) + P_2(a^{-1})a^n + \left( P_4(a) + P_3(a)a^n \right) b\right) \\ &= P_1((-1)^{-1})-P_2((-1)^{-1}) + \sqrt{-1}(P_4(-1)-P_3(-1)) \\
&= P_1(-1)-P_2(-1) - \sqrt{-1}(P_3(-1)-P_4(-1))
\end{align*}

\noindent At this point, we recall that $(\sqrt{-1})^q=-\sqrt{-1}$, since $\sqrt{-1}\in \Fqr{2}\smallsetminus \Fq$ by Lemma~\ref{raiz-1} and $\Fqr{2}^{\times}$ has only two elements of order $4$. Now we compute the right side of the equality $\psi(\hat{u}) = \nu(\psi(u))$  for $j=1$.
\begin{align*}
\phi_1(\eta_1(u)) &= \phi_1\left(P_1(-1)-P_2(-1) + \sqrt{-1}(P_3(-1)-P_4(-1))\right) \\
&= \left(P_1(-1)-P_2(-1)\right)^q + (\sqrt{-1})^q(P_3(-1)-P_4(-1))^q\\
&= P_1(-1)-P_2(-1) - \sqrt{-1}(P_3(-1)-P_4(-1))
\end{align*}

\noindent \underline{$2\leq j \leq t$.} If $4\, \mid \, \deg(g_j)$, then  by Theorem \ref{theo_quat_GaoYue} v), the equality $\psi(\hat{u}) = \nu(\psi(u))$ remains as follows.

\begin{align*}
\eta_j(\hat{u}) &= \eta_j\left(P_1(a^{-1}) + P_2(a^{-1})a^n + \left( P_4(a) + P_3(a)a^n \right) b\right) \\
&= \gamma_j \left( \begin{bmatrix} P_1(\beta_j^{-1}) - P_2(\beta_j^{-1}) & \sqrt{-1}\left(P_4(\beta_j) - P_3(\beta_j)\right) \\ \sqrt{-1}\left(P_4(\beta_j^{-1}) - P_3(\beta_j^{-1})\right) & P_1(\beta_j) - P_2(\beta_j) \end{bmatrix} \right)
\end{align*}

\begin{align*}
\phi_j(\eta_j(u)) &= \phi_j \left( \gamma_j \left( \begin{bmatrix} P_1(\beta_j) - P_2(\beta_j) & \sqrt{-1}\left(P_3(\beta_j) - P_4(\beta_j)\right) \\ \sqrt{-1}\left(P_3(\beta_j^{-1}) - P_4(\beta_j^{-1})\right) & P_1(\beta_j^{-1}) - P_2(\beta_j^{-1}) \end{bmatrix} \right) \right)\\
&= \left( \gamma_j \circ \mathcal{S} \right) \left(\begin{bmatrix} P_1(\beta_j) - P_2(\beta_j) & \sqrt{-1}\left(P_3(\beta_j) - P_4(\beta_j)\right) \\ \sqrt{-1}\left(P_3(\beta_j^{-1}) - P_4(\beta_j^{-1})\right) & P_1(\beta_j^{-1}) - P_2(\beta_j^{-1}) \end{bmatrix} \right) \\
&= \gamma_j \left( \begin{bmatrix} P_1(\beta_j^{-1}) - P_2(\beta_j^{-1}) & \sqrt{-1}\left(P_4(\beta_j) - P_3(\beta_j)\right) \\ \sqrt{-1}\left(P_4(\beta_j^{-1}) - P_3(\beta_j^{-1})\right) & P_1(\beta_j) - P_2(\beta_j) \end{bmatrix} \right)
\end{align*}

\noindent It can be analogously proved the equality whenever $4\, \nmid \, \deg(g_j)$. Finally we have:

\medskip

\noindent \underline{$t+1\leq j \leq t+k$.}
\begin{align*}
\eta_j(\hat{u}) &= \eta_j\left( P_1(a^{-1}) + P_2(a^{-1})a^n + \left( P_4(a) + P_3(a)a^n \right) b \right) \\
&= \begin{bmatrix} P_1(\beta_j^{-1}) - P_2(\beta_j^{-1}) & P_3(\beta_j) - P_4(\beta_j) \\ P_4(\beta_j^{-1}) - P_3(\beta_j^{-1}) & P_1(\beta_j) - P_2(\beta_j) \end{bmatrix}
\end{align*}

\begin{align*}
\phi_j(\eta_j(u)) &= \phi_j \left( \begin{bmatrix} P_1(\beta_j) - P_2(\beta_j) & P_4(\beta_j) - P_3(\beta_j) \\[2mm] P_3(\beta_j^{-1}) - P_4(\beta_j^{-1}) & P_1(\beta_j^{-1}) - P_2(\beta_j^{-1}) \end{bmatrix} \right) \\
&= \mathcal{S} \left(\begin{bmatrix} P_1(\beta_j) - P_2(\beta_j) & P_4(\beta_j) - P_3(\beta_j) \\[2mm] P_3(\beta_j^{-1}) - P_4(\beta_j^{-1}) & P_1(\beta_j^{-1}) - P_2(\beta_j^{-1}) \end{bmatrix} \right) \\
&= \begin{bmatrix} P_1(\beta_j^{-1}) - P_2(\beta_j^{-1}) & P_3(\beta_j) - P_4(\beta_j) \\ P_4(\beta_j^{-1}) - P_3(\beta_j^{-1}) & P_1(\beta_j) - P_2(\beta_j) \end{bmatrix}
\end{align*}
\end{proof}

\medskip

Now it is enough to apply $\nu$ in the previous result to the  descomposition of the corresponding ideal in $\Fq[Q_n]$ of a $Q_n$-code to obtain its Euclidean dual. This is the main result of this section.

\begin{theorem}
\label{quat_euclid_orthogonal}
Assume \emph{Hypotheses~\ref{hipotesis_descomposicion_FqQn}}.  Let $\mathcal{C}$ be a $Q_n$-code  over $\Fq$, and $I_{\cC}$ be its corresponding ideal in $\mathbb{F}_q[Q_n]$.  Consider the decomposition of $I_{\cC}$  through the isomorphism $\psi$ given in \emph{Theorem~\ref{theo_quat_GaoYue}}, that is
$$
I_{\cC}\overset{\psi}{\cong} \displaystyle \bigoplus_{i=1}^{r+s} I_i 
\oplus \displaystyle\bigoplus_{j=1}^{t+k} L_j 
 \subseteq \displaystyle \bigoplus_{i=1}^{r+s} A_i 
\oplus \displaystyle\bigoplus_{j=1}^{t+k} B_j .
$$
Then $I_{\cC^{\perp_E}}\subseteq \mathbb{F}_q[Q_n]$ verifies that 
$$
I_{\cC^{\perp_E}}\overset{\psi}{\cong} \displaystyle \bigoplus_{i=1}^{r+s} I_i^{\perp_E} 
\oplus \displaystyle \bigoplus_{j=1}^{t+k} L_j^{\perp_E}, 
$$ 
where all $I_i^{\perp_E}$ 
are given as in \emph{Theorem~\ref{dih_Euclidean_dual}}, and for $1\leq j\leq t+k$ the ideals $L_j^{\perp_E}\subseteq B_j$ are as follows:  
$$
L_j^{\perp_E} =
\left\lbrace\begin{array}{ll}
B_j, & \text{if } L_j=\pmb{0} \\*[4mm]
\pmb{0}, & \text{if } L_j=B_j \\*[4mm]
L_j, & \text{otherwise}
\end{array}\right. 
.$$
\end{theorem}

\begin{proof}
Notice that the ``dihedral part'' of the statement is already proved in Theorem~\ref{dih_Euclidean_dual}, so we may focus only on $1\leq j \leq t+k$. In virtue of (\ref{eq:orthogonal_ideal_Euclidean}) and Theorems~\ref{theo_commutative_quaternion} and~\ref{theo_quat_GaoYue}, it is enough to apply $\nu$ (more concretely, $\phi_j$) to $\operatorname{Ann}_r(L_j)$ for each possible ideal $L_j$ of $B_j$.
Note that, if $L_j=\langle N_j \rangle$, where $N_j\in B_j$ is a matrix in row reduced echelon form, then $\operatorname{Ann}_r(L_j)$ is just 
\[
\operatorname{Ann}_r(L_j)=\operatorname{Ann}_r(\langle N_j\rangle)=\{Y \in B_j\ |\ N_jY=0\}.
\]
We argue case by case.

\medskip

\noindent \underline{$j=1$.} In this case $B_1=\Fq(\sqrt{-1})$, so  the only possible ideals are $\pmb{0}$ and the whole $B_1$. Certainly the right annihilator of $\pmb{0}$ is $B_1$ and vice versa. Hence, applying $\phi_1:x\mapsto x^q$ yields the desired conclusion in this case.

\bigskip

\noindent \underline{$2\leq j \leq t$ with $4\mid \deg(g_j)$.} The possibilities for the matrix generator $N_j$ of the ideal $L_j\subseteq B_j$, in row reduced echelon form, are $$\begin{bmatrix}
0&0\\0&0
\end{bmatrix}, \qquad \begin{bmatrix}
0&1\\0&0
\end{bmatrix}, \qquad \begin{bmatrix}
1 & \lambda_j \\ 0&0 
\end{bmatrix} \qquad \text{and} \qquad \begin{bmatrix}
1&0\\0&1
\end{bmatrix},$$ where  $\lambda_j\in\Fq(\beta_j+\beta_j^{-1})$. It is easy to check that $\operatorname{Ann}_r(L_j)$ can be generated en each case, respectively, by $$\begin{bmatrix}
1&0\\0&1
\end{bmatrix}, \qquad \begin{bmatrix}
0&1\\0&0
\end{bmatrix}, \qquad \begin{bmatrix}
0 & -\lambda_j \\ 0&1
\end{bmatrix} \qquad \text{and} \qquad \begin{bmatrix}
0&0\\0&0
\end{bmatrix}.$$ 
Finally, since $L_j^{\perp_E}=\phi_j(\operatorname{Ann}_r(L_j))$, we should apply $\phi_j:X\mapsto (\gamma_j\circ \mathcal{S}\circ \gamma_j^{-1})(X)$ to them in order to get the generators of $L_j^{\perp_E}$.
Clearly, it is enough to make the computations for the proper ideals. First, using Lemma~\ref{antiautomorphisms}, we compute
\begin{eqnarray*}
( \gamma_j\circ \mathcal{S}\circ \gamma_j^{-1}) \left( \begin{bmatrix}
0&1\\0&0 \end{bmatrix}\right) &= &  \gamma_j\circ \mathcal{S}\left(Z_j\begin{bmatrix}
0&1\\ 0&0 
\end{bmatrix}Z_j^{-1} \right)  \\
& = & \gamma_j\left(\mathcal{S}(Z_j)^{-1}\begin{bmatrix}
0&-1\\ 0&0 
\end{bmatrix}\mathcal{S}(Z_j)\right) \\
& = & Z_j^{-1}\mathcal{S}(Z_j)^{-1}\begin{bmatrix}
0&-1\\ 0&0 
\end{bmatrix}\mathcal{S}(Z_j)Z_j.
\end{eqnarray*}
Hence if  $L_j = \langle \left[\begin{smallmatrix} 0&1\\0&0\end{smallmatrix}\right]\rangle$, then $L_j^{\perp_E}$ 
is generated by the above matrix product. But $L_j^{\perp_E}$ is a left ideal, so to obtain a generator it is enough to compute
\begin{eqnarray*}
\begin{bmatrix}
0&-1\\ 0&0 
\end{bmatrix}\mathcal{S}(Z_j)Z_j &=& \begin{bmatrix}
0&-1\\ 0&0 
\end{bmatrix}\begin{bmatrix}
-\beta_j^{-1}&\beta_j\\-\sqrt{-1}&\sqrt{-1}
\end{bmatrix}\begin{bmatrix}
\sqrt{-1}&-\beta_j\\\sqrt{-1}&-\beta_j^{-1}
\end{bmatrix} \\
& = & \begin{bmatrix}
\sqrt{-1}&-\sqrt{-1}\\0&0
\end{bmatrix}\begin{bmatrix}
\sqrt{-1}&-\beta_j\\\sqrt{-1}&-\beta_j^{-1}
\end{bmatrix} \\
& = & \begin{bmatrix}
0&-\sqrt{-1}(\beta_j-\beta_j^{-1})\\ 0&0
\end{bmatrix}.
\end{eqnarray*}
The row reduced echelon form of this last matrix is $\left[\begin{smallmatrix} 0&1\\0&0\end{smallmatrix}\right]$, because $2\leq j \leq t$ and so $\beta_j\neq \beta_j^{-1}$. We have thus deduced  $L_j=L_j^{\perp_E}$ when $L_j = \langle \left[\begin{smallmatrix} 0&1\\0&0\end{smallmatrix}\right]\rangle$. Second, we have 
\begin{eqnarray*}
 (\gamma_j\circ \mathcal{S}\circ \gamma_j^{-1}) \left( \begin{bmatrix}
0&-\lambda_j\\0&1 \end{bmatrix}\right) &= &  \gamma_j\circ \mathcal{S}\left(Z_j\begin{bmatrix}
0&-\lambda_j\\ 0&1 
\end{bmatrix}Z_j^{-1} \right)  \\
& = & Z_j^{-1}\mathcal{S}(Z_j)^{-1}\begin{bmatrix}
1&\lambda_j\\ 0&0
\end{bmatrix}\mathcal{S}(Z_j)Z_j.
\end{eqnarray*}
Moreover
\begin{eqnarray*}
\begin{bmatrix}
1&\lambda_j\\ 0&0
\end{bmatrix}\mathcal{S}(Z_j)Z_j &=& \begin{bmatrix}
1&\lambda_j\\ 0&0
\end{bmatrix}\begin{bmatrix}
-\beta_j^{-1}&\beta_j\\-\sqrt{-1}&\sqrt{-1}
\end{bmatrix}\begin{bmatrix}
\sqrt{-1}&-\beta_j\\\sqrt{-1}&-\beta_j^{-1}
\end{bmatrix} \\
& = & \begin{bmatrix}
-\beta_j^{-1}-\lambda_j\sqrt{-1}& \beta_j+\lambda_j\sqrt{-1} \\ 0&0
\end{bmatrix}\begin{bmatrix}
\sqrt{-1}&-\beta_j\\\sqrt{-1}&-\beta_j^{-1}
\end{bmatrix} \\
& = & \begin{bmatrix}
\sqrt{-1}(\beta_j-\beta_j^{-1})&\lambda_j\sqrt{-1}(\beta_j-\beta_j^{-1})&\\ 0&0
\end{bmatrix},
\end{eqnarray*}
whose row reduced echelon form is $\left[\begin{smallmatrix}
1 & \lambda_j \\ 0&0 
\end{smallmatrix}\right]$. It follows    $L_j=L_j^{\perp_E}$ when $L_j = \langle \left[\begin{smallmatrix} 1 & \lambda_j \\ 0&0\end{smallmatrix}\right]\rangle$.

\bigskip

\noindent \underline{$2\leq j \leq t$ with $4\nmid \deg(g_j)$.} Here we argue as in the previous situation but applying $\phi_j:X\mapsto (\theta_j\circ \mathcal{S}\circ \theta_j^{-1})(X)$. Therefore, basically, we must compute 
\begin{equation}\label{coordinates}
(\theta_j\circ \mathcal{S}\circ \theta_j^{-1}) \left( \begin{bmatrix}
0&1\\0&0 \end{bmatrix}\right) \qquad \text{and} \qquad (\theta_j\circ \mathcal{S}\circ \theta_j^{-1}) \left( \begin{bmatrix}
0&-\lambda_j\\0&1 \end{bmatrix}\right).
\end{equation}

From the second expression of $S_j$ given in Lemma~\ref{lemma:theta} (a), it is easy to see that the next set is a basis of $S_j$ over $\Fq(\beta_j+\beta_j^{-1})$:
$$
\mathcal{B}_1:=\left\lbrace \begin{bmatrix} 1&0\\0&1\end{bmatrix}, \begin{bmatrix} 0&1\\-1&0\end{bmatrix}, \begin{bmatrix} \sqrt{-1}&0\\0&-\sqrt{-1}\end{bmatrix}, \begin{bmatrix} 0&\sqrt{-1}\\\sqrt{-1}&0\end{bmatrix}\right\rbrace.
$$ 
Moreover, following the definition of $\theta_j$ in Lemma~\ref{lemma:theta} (b), we see that 
$$
\mathcal{B}_2:=\left\lbrace \begin{bmatrix} 1&0\\0&1\end{bmatrix}, \begin{bmatrix} 0&1\\-1&0\end{bmatrix}, \begin{bmatrix} u&v\\v&-u\end{bmatrix}, \begin{bmatrix} -v&u\\u&v\end{bmatrix} \right\rbrace
$$
is clearly a basis of ${\cal M}_2(\Fq(\beta_j+\beta_j^{-1})$. It follows that, if the coordinates of $X\in {\cal M}_2(\Fq(\beta_j+\beta_j^{-1})$ in the basis $\mathcal{B}_2$ are $(a_1,a_2,a_3,a_4)$, then these are also the coordinates of $\theta_j^{-1} (X)$ in $\mathcal{B}_1$. Thus, since ${\cal S}$ is an homomorphism, it easy to check that the coordinates of $(\mathcal{S}\circ \theta_j^{-1}) (X)$ in $\mathcal{B}_1$ are $(a_1,-a_2,-a_3,-a_4)$, which in addition coincide with the coordinates of $(\theta_j\circ \mathcal{S}\circ \theta_j^{-1}) (X)$ in $\mathcal{B}_2$.  We use this argument to compute $(\ref{coordinates})$.

One can easily check that the coordinates of $\left[\begin{smallmatrix} 0&1\\0&0\end{smallmatrix}\right]$ in $\mathcal{B}_2$ are $(0,\tfrac{1}{2}, \tfrac{-v}{2}, \tfrac{-u}{2})$, so its image by $\theta_j\circ \mathcal{S}\circ \theta_j^{-1}$ is 
$$
(\theta_j\circ \mathcal{S}\circ \theta_j^{-1}) \left( \begin{bmatrix}
0&1\\0&0 \end{bmatrix}\right) = -\frac{1}{2} \begin{bmatrix} 0&1\\-1&0\end{bmatrix}  +\frac{v}{2}\begin{bmatrix} u&v\\v&-u\end{bmatrix} + \frac{u}{2} \begin{bmatrix} -v&u\\u&v\end{bmatrix} = \begin{bmatrix}  0 & -1 \\ 0 & 0  \end{bmatrix}.
$$Consequently in this case it also holds $L_j=L_j^{\perp_E}$ when $L_j = \langle \left[\begin{smallmatrix} 0 & 1 \\ 0&0\end{smallmatrix}\right]\rangle$.

Finally, the coordinates of $\left[\begin{smallmatrix} 0& -\lambda_j\\0&1\end{smallmatrix}\right]$ in $\mathcal{B}_2$ are $\tfrac{1}{2}(1,-\lambda_j,u+v\lambda_j, u\lambda_j-v)$,  so 
\begin{eqnarray*}
(\theta_j\circ \mathcal{S}\circ \theta_j^{-1}) \left( \begin{bmatrix}
0&-\lambda_j\\0&1 \end{bmatrix}\right) &= & \frac{1}{2}\begin{bmatrix} 1&0\\0&1\end{bmatrix} +\frac{\lambda_j}{2} \begin{bmatrix} 0&1\\-1&0\end{bmatrix}  -\frac{u+v\lambda_j}{2}\begin{bmatrix} u&v\\v&-u\end{bmatrix} + \frac{v-u\lambda_j}{2} \begin{bmatrix} -v&u\\u&v\end{bmatrix} \\
&=& \begin{bmatrix}  1 & \lambda_j \\ 0 & 0  \end{bmatrix}, \\
\end{eqnarray*}
and it also follows $L_j=L_j^{\perp_E}$ when $L_j=\langle \left[\begin{smallmatrix} 1 & \lambda_j \\ 0&0\end{smallmatrix}\right]\rangle $.   


\bigskip

\noindent \underline{$t+1\leq j \leq t+k$.} In this last case $\phi_j:X\mapsto \mathcal{S}(X)$, and straightforward computations yield the desired conclusion.
\end{proof}

\medskip

As an application of Theorem~\ref{quat_euclid_orthogonal}, we can enumerate Euclidean self-orthogonal $Q_n$-codes over $\Fq$. In doing so, we have unfortunately detected a mistake in \cite[Theorem 4.5]{GaoYue21}  where it is claimed that this number is always a power of $3$; this will be further illustrated in Example \ref{ex-error}.

\begin{corollary}
\label{cor_quat_self}
Assume \emph{Hypotheses~\ref{hipotesis_descomposicion_FqQn}}. Following the notation in \emph{Theorem~\ref{quat_euclid_orthogonal}}, consider the next expression of an ideal $I_{\cC}\subseteq\Fq[Q_n]$ associated with a $Q_n$-code $\mathcal{C}$ over $\Fq$: 
$$
I_{\cC}\overset{\psi}{\cong} \displaystyle \bigoplus_{i=1}^{r+s} I_i 
\oplus \displaystyle\bigoplus_{j=1}^{t+k} L_j 
 \subseteq \displaystyle \bigoplus_{i=1}^{r+s} A_i 
\oplus \displaystyle\bigoplus_{j=1}^{t+k} B_j 
$$
Then $\mathcal{C}\subseteq \mathcal{C}^{\perp_E}$ if and only if the following conditions are satisfied:
\begin{enumerate}[label=\emph{\roman*)}]
\setlength{\itemsep}{0mm}
\item  $I_1= \pmb{0}\oplus \pmb{0}$.
\item If $2\leq i \leq r$, then $I_i=\pmb{0}$.
\item If $r+1\leq i \leq r+s$, then  $I_i\in \{ \pmb{0} , \langle  \left[\begin{smallmatrix} 0&1 
\\0&0\end{smallmatrix}\right]\rangle, \langle  \left[\begin{smallmatrix} 1&0\\0&0\end{smallmatrix}\right]\rangle \}$.
\item If $1\leq j\leq t+k$, then  $L_j$ 
can be any ideal strictly contained in $B_j$.
\end{enumerate}
In particular, there are exactly $3^{s}\displaystyle\prod_{j=2}^{t}(q^{\deg(g_j)/2}+2)\displaystyle\prod_{j=t+1}^{t+k}(q^{\deg(g_j)}+2)$ Euclidean self-orthogonal $Q_n$-codes over $\Fq$.
\end{corollary}

\begin{proof}
It is enough to check which ideals we can take in Theorem~\ref{quat_euclid_orthogonal} (see also Theorem~\ref{dih_Euclidean_dual}) such that they are contained in the corresponding Euclidean orthogonal, for each index $i$ and $j$:
\begin{itemize}
	\item For $i=1$ and for $r+1\leq i \leq r+s$, the claim is clear from Theorem~\ref{dih_Euclidean_dual}. 
	
	\item For $2\leq i \leq r$, we claim that the unique possibility is  $I_i=\pmb{0}$;  
and for that aim, it is clearly enough to show that 
$$
\begin{bmatrix} 1&\lambda_i\\0&0\end{bmatrix}\notin \left\langle \begin{bmatrix} \alpha_i+\alpha_i^{-1}+2\lambda_i &-2-(\alpha_i+\alpha_i^{-1})\lambda_i\\0&0 \end{bmatrix}\right\rangle_{A_i}.
$$ Arguing by contradiction, then $-2-(\alpha_i+\alpha_i^{-1})\lambda_i=\lambda_i(\alpha_i+\alpha_i^{-1}+2\lambda_i)$, so $\lambda_i^2+(\alpha_i+\alpha_i^{-1})\lambda_i+1=0$. But the roots of the polynomial $\mx^2+(\alpha_i+\alpha_i^{-1})\mx+1\in\Fq(\alpha_i+\alpha_i^{-1})[\mx]$ are $-\alpha_i, -\alpha_i^{-1}\notin \Fq(\alpha_i+\alpha_i^{-1})$, a contradiction.
	
	\item Finally, for $1\leq j \leq t+k$, the assertion directly follows from Theorem~\ref{quat_euclid_orthogonal}. 
	\end{itemize}
As a consequence of i)-iv), we can now derive the formula regarding the number of Euclidean self-orthogonal $Q_n$-codes over $\Fq$. Statements i) and ii) give one choice only. Statement iii) provides $3^s$ different choices. In statement iv) we distinguish three cases. For $j=1$ we have a unique ideal strictly contained in $B_1=\Fq(\sqrt{-1})$, so again a unique possibility. For $2\leq j \leq t$ it holds that $B_j\cong\mathcal{M}_2(\Fqr{\deg(g_j)/2})$ contains $q^{\deg(g_j)/2} + 1$ left ideals of rank 1, and together with the zero ideal we have $q^{\deg(g_j)/2} + 2$ different choices for $L_j$. An analogous argument applies when $t+1\leq j \leq t+k$.
\end{proof}

\smallskip

\begin{remark}
It is also possible to provide, from Theorem~\ref{quat_euclid_orthogonal}, the full description of the Euclidean hull of any $Q_n$-code over $\Fq$ via its decomposition within the semisimple group algebra $\Fq[Q_n]$. In particular, all Euclidean LCD $Q_n$-codes over $\Fq$ can be computed and enumerated.
\end{remark}


\section{Some numerical examples}
\label{sec:examples}

This final section is devoted to illustrate how the theoretical results stated in this paper can be applied. In the first examples that follow, we will construct some  Hermitian self-orthogonal dihedral codes arising from the group algebra $\Fqq[D_n]$. Then we will apply the Hermitian CSS construction (see subsection~\ref{pre:CSS}) to obtain quantum-error correcting codes which achieve the best distance for their length and dimension. Additionally, in the last two examples we illustrate the incorrect formulae appearing in \cite[Theorem 4.3]{CaoCaoFu23} and \cite[Theorem 4.5]{GaoYue21}, respectively, on the number of Hermitian self-orthogonal $D_n$-codes over $\Fqq$ and on the number of Euclidean self-orthogonal $Q_n$-codes over $\Fq$.

\begin{example}
Let $q=3$ and $n=16$. Consider the group algebra $\F_{9}[D_{16}]$, where $D_{16}=\langle a,b \: | \: a^{16}=b^2=1, bab=a^{-1}\rangle$ is a dihedral group of order $32$. Let $\omega$ be a primitive element of $\F_9$. Since $\car(\F_9)=3$ does not divide $n=16$, according to Theorem~\ref{theo_dih_refined}, the Wedderburn-Artin decomposition of $\F_{9}[D_{16}]$ depends on the factorisation of $\mx^{16}-1\in\F_9[\mx]$ into irreducible polynomials: 
$$
\mx^{16}-1=\underset{f_1}{(\mx-1)}\underset{f_2}{(\mx+1)}\underset{f_3}{(\mx+\omega)}\underset{f_4}{(\mx+\omega^2)}\underset{\overline{f_3}}{(\mx+\omega^3)}\underset{f_3^{\dagger}}{(\mx+\omega^5)}\underset{f_4^{*}=\overline{f_4}}{(\mx+\omega^6)}\underset{f_3^{*}}{(\mx+\omega^7)}\cdot$$ 
$$\cdot\underset{f_5}{(\mx^2+\omega)}\underset{\overline{f_5}}{(\mx^2+\omega^3)}\underset{f_5^{\dagger}}{(\mx^2+\omega^5)}\underset{f_5^{*}}{(\mx^2+\omega^7)}.  
$$
Hence, following the notation in Hypotheses~\ref{hipotesis_descomposicion_Fq2Dn}, we get $J_0=\{1,2\}, J_1=J_2=\emptyset, J_3=\{4\}$ and $J_4=\{3,5\}$. Now by Theorem~\ref{theo_dih_refined} and Remark \ref{remark_fields} we have that $\F_{9}[D_{16}]$ can be decomposed as 
$$\F_{9}[D_{16}] \overset{\rho}{\cong} \underset{A_1}{(\F_9 \oplus \F_9)} \oplus \underset{A_2}{(\F_9 \oplus \F_9)} \oplus \underset{A_3}{(\mathcal{M}_2(\F_9)\oplus \mathcal{M}_2(\F_9))} \oplus \underset{A_4}{\mathcal{M}_2(\F_9)}\oplus \underset{A_5}{(\mathcal{M}_2(\F_{9^2})\oplus \mathcal{M}_2(\F_{9^2}))}, $$
and $\rho$ is given by the generators of $D_{16}$ as follows:
$$ \rho(a) = \left(1, 1 , -1 , -1 ,  \begin{pmatrix} \omega^5 & 0 \\ 0& \omega^3 \end{pmatrix} , \begin{pmatrix} \omega^7 & 0 \\ 0& \omega\end{pmatrix}  , \begin{pmatrix} \omega^6 & 0\\ 0 & \omega^2 \end{pmatrix} ,  \begin{pmatrix} \xi^{25} & 0 \\ 0& \xi^{55} \end{pmatrix} , \begin{pmatrix} \xi^{75} & 0 \\ 0& \xi^5\end{pmatrix}\right),$$ where $\xi$ denotes a primitive element of $\F_{9^2}$, and 
$$ \rho(b) = \left(1, -1 , 1 , -1 , \begin{pmatrix} 0&1\\ 1 & 0 \end{pmatrix} , \begin{pmatrix} 0&1\\ 1 & 0 \end{pmatrix} , \begin{pmatrix} 0&1\\ 1 & 0\end{pmatrix}  , \begin{pmatrix} 0&1\\ 1 & 0 \end{pmatrix} , \begin{pmatrix} 0&1\\ 1 & 0\end{pmatrix} \right).$$ 
Consider now the $D_{16}$-code $\cC$ over $\F_9$ such that 
$$ I_{\cC} \cong \pmb{0} \oplus \pmb{0} \oplus \pmb{0} \oplus \pmb{0} \oplus \pmb{0} \oplus \left\langle \begin{pmatrix} 1& 1 \\ 0& 0\end{pmatrix}\right\rangle \oplus \left\langle \begin{pmatrix} 1& \omega^7 \\ 0& 0\end{pmatrix}\right\rangle\oplus \left\langle \begin{pmatrix} 1& \xi^{14} \\ 0& 0\end{pmatrix}\right\rangle \oplus \left\langle \begin{pmatrix} 1& \xi^2 \\ 0& 0\end{pmatrix}\right\rangle.$$ The dimension of $\cC$ can be easily computed via (\ref{eq:dimension}), and indeed $\dim_{\F_9}(\cC)= 12$. Using the software \begin{sc}Magma\end{sc} \cite{magma}, we were able to construct this concrete $\rho$ and its inverse, and so we computed the next generator element of $I_{\cC}$ in $\F_9[D_{16}]$: 
$$
\omega^3 a + \omega^7 a^2 + \omega a^3 + a^5 + \omega^5 a^6 + a^7 + a^8 + \omega^3 a^9 + \omega^2 a^{10} + \omega a^{11} + a^{13} + \omega^6 a^{14} + a^{15} + \omega b + ab + \omega a^2b$$
$$ + \omega^2 a^3b + \omega^6 a^4b + \omega^2 a^5b - a^6b + \omega a^7b + \omega a^8b + \omega^7 a^9b + \omega^3 a^{10}b + \omega^2 a^{11}b + \omega^6 a^{13}b + \omega^3 a^{14}b + \omega^3 a^{15}b.$$
Moreover, we obtained its minimum distance, which is 12, that is, $\cC$ is a $D_{16}$-code over $\F_9$ with parameters $[32,12,12]_9$. Notice that this code is Hermitian self-orthogonal by Theorem~\ref{dih-self-hermitic}; therefore, in virtue of the CSS construction (see Theorem \ref{theo:CSS} of subsection~\ref{pre:CSS}), and making computations with \begin{sc}Magma\end{sc}, we obtain a quantum $D_{16}$-code $\mathcal{Q}$ with parameters $[[32,8,8]]_3$. We point out that $\mathcal{Q}$ attains the best known minimum distance for these parameters according to Grassl's tables \cite{GrasslTable}. 
\end{example}

\smallskip

\begin{example}
Consider again $\F_9[D_{16}]$, as in the previous example, and consider on this occassion the $D_{16}$-code over $\F_9$ associated with 
$$ 
I_{\cC} \cong \pmb{0} \oplus \pmb{0} \oplus \pmb{0} \oplus \pmb{0} \oplus \left\langle \begin{pmatrix} 1& \omega^4 \\ 0& 0\end{pmatrix}\right\rangle \oplus \pmb{0} \oplus \left\langle \begin{pmatrix} 1& \omega^7 \\ 0& 0\end{pmatrix}\right\rangle \oplus \pmb{0} \oplus \left\langle \begin{pmatrix} 1& \xi^{23} \\ 0& 0\end{pmatrix}\right\rangle.
$$ 
It follows by (\ref{eq:dimension}) that $\dim_{\F_9}(\cC)=8$. Using \begin{sc}Magma\end{sc} we checked that its parameters are $[32,8,19]_9$, and a generator of $I_{\cC}$ is:
$$1 - a + \omega^2 a^2 - a^3 + \omega^6 a^4 + \omega^5 a^5 + \omega^6 a^6 + \omega^7 a^7 - a^9 + a^{10} - a^{11} + \omega^2 a^{12} + \omega^5 a^{13} + a^{14} + \omega^7 a^{15} $$ 
$$- b - ab + a^2b + \omega^3 a^3b + \omega^6 a^5b - a^6b - a^7b + \omega^5 a^8b + \omega^7 a^9b - a^{11}b + \omega^5 a^{12}b + \omega a^{13}b + \omega^3 a^{14}b + a^{15}b.$$
This code is also Hermitian self-orthogonal by Theorem~\ref{dih-self-hermitic}, so one can construct with \begin{sc}Magma\end{sc} the associated quantum code $\mathcal{Q}$. In this case $\mathcal{Q}$ has as parameters $[[32,16,6]]_3$, which achieves the best known minimum distance according to \cite{GrasslTable}.

Similarly, one may construct many other quantum dihedral codes over $\F_3$ that attain the best known minimum distance. For instance, we can build one with parameters $[[20,12,4]]_3$, through the Hermitian self-orthogonal $D_{10}$-code over $\F_9$ corresponding to the ideal in $\F_9[D_{10}]$ generated by the element $$
1 + \omega^2 a + \omega a^2 + a^5 + \omega^7 a^6 + \omega^5 a^7 + \omega a^8 + a^9 + b + a b + \omega a^2 b + \omega^5 a^3 b + \omega^7 a^4 b + a^5 b + \omega a^8 b + \omega^2 a^9 b.
$$
\end{example}

\smallskip

\begin{example}
\label{ex_cao}
This example aims to illustrate how our alternative approach, via the group algebra structure, to construct and enumerate dihedral (Hermitian self-orthogonal) codes reduces the computations given in \cite{CaoCaoFu23}; compare this example with \cite[Example 5.1]{CaoCaoFu23}.

Take $q=2$ and $n=7$, and consider the group algebra $\F_4[D_7]$. Note that $\mx^7-1=\mx^7+1\in\F_4[\mx]$ can be factorised into irreducible polynomials as follows: $$\mx^7+1 = \underset{f_1}{(\mx+1)}\underset{f_2=\overline{f_2}}{(\mx^3+\mx+1)}\underset{f_2^{*}}{(\mx^3+\mx^2+1)}.$$ According to the notation in Hypotheses~\ref{hipotesis_descomposicion_Fq2Dn}, we thus have $J_0=\{1\}$, $J_2=\{2\}$ and $J_1=J_3=J_4=\emptyset$; moreover, by Theorem~\ref{theo_dih_refined}, the Wedderburn-Artin decomposition of $\F_4[D_7]$ is: 
$$
\F_4[D_7]\cong \underset{A_1}{\F_4[C_2]} \oplus \underset{A_2}{\mathcal{M}_2(\F_{4^3})},
$$ 
where $C_2$ denotes the cyclic group generated by $\left[\begin{smallmatrix} 0&1\\1&0\end{smallmatrix}\right]$. Observe that there are only three ideals in $A_1=\F_4[C_2]$, being  $\langle\left[\begin{smallmatrix} 1&1\\1&1\end{smallmatrix}\right]\rangle_{\F_{4}[C_{2}]}$ the unique proper ideal. Moreover, each ideal of $A_2=\mathcal{M}_2(\F_{4^3})$ is generated by a matrix in row reduced echelon form, and there are $4^3+3=67$ of them. Consequently, there are $3\cdot 67=201$ possible $D_7$-codes over $\F_4$, which are isomorphic to 
$$ 
\langle M_1\rangle_{A_1} \oplus \langle M_2\rangle_{A_2},
$$ 
where $M_1\in \{ \left[\begin{smallmatrix} 0&0\\0&0\end{smallmatrix}\right], \left[\begin{smallmatrix} 1&1\\1&1\end{smallmatrix}\right],\left[\begin{smallmatrix} 1&0\\0&1\end{smallmatrix}\right]\}$ and $M_2\in \{ \left[\begin{smallmatrix} 0&0\\0&0\end{smallmatrix}\right],\left[\begin{smallmatrix} 0&1\\0&0\end{smallmatrix}\right],\left[\begin{smallmatrix} 1&\lambda\\0&0\end{smallmatrix}\right],\left[\begin{smallmatrix} 1&0\\0&1\end{smallmatrix}\right]\}$ with $\lambda\in \F_{4^3}$. In particular, by Corollary~\ref{cor_dih_sel}, twenty of them are Hermitian self-orthogonal codes.
\end{example}

\smallskip

\begin{example}
\label{example-caocao}
Consider on this occasion the group algebra $\F_4[D_9]$, so $q=2$ and $n=9$. We can factorise the polynomial $\mx^9-1=\mx^9+1\in\F_4[\mx]$ into irreducible factors as follows, where $\omega$ is a primitive element of $\F_4$: $$\mx^9+1 = \underset{f_1}{(\mx+1)}\underset{f_2}{(\mx+\omega)}\underset{f_2^{*}=\overline{f_2}}{(\mx+\omega^2)}\underset{f_3}{(\mx^3+\omega)}\underset{f_3^{*}=\overline{f_3}}{(\mx^3+\omega^2)}.$$ Following the notation in Hypotheses~\ref{hipotesis_descomposicion_Fq2Dn} we have $J_0=\{1\}$, $J_3=\{2,3\}$ and $J_1=J_2=J_4=\emptyset$. In particular, by Theorem~\ref{theo_dih_refined} we deduce \begin{equation}\label{eq:ex} \F_4[D_9]\cong \underset{A_1}{\F_4[C_2]} \oplus \underset{A_2}{\mathcal{M}_2(\F_{4})} \oplus \underset{A_3}{\mathcal{M}_2(\F_{4^3})},\end{equation} where $C_2$ denotes the cyclic group generated by $\left[\begin{smallmatrix} 0&1\\1&0\end{smallmatrix}\right]$. Now Corollary~\ref{cor_dih_sel} ensures that in this group algebra there are $2\cdot (2+2)\cdot(2^3+2)=80$ Hermitian self-orthogonal $D_{9}$-codes over $\F_4$, contrary to the formula in \cite[Theorem 4.3]{CaoCaoFu23} which asserts the amount $32$ of such codes. In fact, our computations with \textsc{Magma} together with Theorem~\ref{dih-self-hermitic}~v) and Proposition~\ref{prop_solutions}~(c) confirm that all the ideals of (\ref{eq:ex}) with the following structures correspond to Hermitian self-orthogonal codes: 
$$ 
\pmb{0} \oplus \left\langle \begin{pmatrix}1& \omega^{k_1} \\ 0&0\end{pmatrix} \right\rangle_{A_{2}}\oplus  \left\langle \begin{pmatrix}1& \xi^{7k_2} \\ 0&0\end{pmatrix} \right\rangle_{A_{3}} \qquad \text{and} \qquad \left\langle \begin{pmatrix}1&1\\1&1 \end{pmatrix}\right\rangle_{A_{1}} \oplus \left\langle \begin{pmatrix}1& \omega^{k_1} \\ 0&0\end{pmatrix} \right\rangle_{A_{2}}\oplus  \left\langle \begin{pmatrix}1& \xi^{7k_2} \\ 0&0\end{pmatrix} \right\rangle_{A_{3}},
$$
for some primitive element $\xi$ of $\F_{4^3}$, for every $k_1\in \{0,1,2\}$ and for every $k_2\in\{0,1,...,8\}$. Certainly there are more than the $32$ claimed in \cite[Theorem 4.3]{CaoCaoFu23}.

\end{example}

\smallskip

\begin{example}
\label{ex-error}
Let $q=11$ and $n=7$, and consider the group algebra $\F_{11}[Q_7]$ of a generalised quaternion group $Q_7=\langle a, b\: | \: a^{14}=1, b^2=a^7, b^{-1}ab=a^{-1}\rangle$ of order $28$. Following the notation in Hypotheses~\ref{hipotesis_descomposicion_FqQn}, since $n$ is odd, we can factorise in $\F_{11}[\mx]$ the polynomials $\mx^7-1$ and $\mx^7+1$ as:
$$ \mx^7-1 = \underset{f_1}{(\mx-1)}\underset{f_2}{(\mx^3+5\mx^2+4\mx-1)}\underset{f_2^{*}}{(\mx^3+7\mx^2+6\mx-1)} ,$$
$$ \mx^7+1 = \underset{g_1}{(\mx+1)}\underset{g_2}{(\mx^3+4\mx^2+6\mx+1)}\underset{g_2^{*}}{(\mx^3+6\mx^2+4\mx+1)} .$$
In particular,  $r=t=1$ and $s=k=1$. Moreover, as $q\equiv 3 \text{ (mod 4)}$, then in virtue of Theorem~\ref{theo_quat_GaoYue} we deduce that \begin{equation}\label{eq:ex2}\F_{11}[Q_7] \overset{\psi}{\cong} \underset{A_1}{(\F_{11}\oplus \F_{11})}\oplus \underset{A_2}{\mathcal{M}_2(\F_{11^3})} \oplus \underset{B_1}{\F_{11}(\sqrt{-1}) }\oplus\underset{B_2}{ \mathcal{M}_2(\F_{11^3})},\end{equation} and $\psi$ is given by the generators of $Q_7$ as follows: 
$$\psi(a) = \left(1 , 1 , \begin{pmatrix} \xi^{570} & 0 \\ 0 & \xi^{760} \end{pmatrix} , -1 , \begin{pmatrix} \xi^{95} & 0 \\ 0 & \xi^{1235} \end{pmatrix}\right),$$ where $\xi$ denotes a primitive element of $\F_{11^3}$, and 
$$\psi(b) = \left(1 , -1 , \begin{pmatrix} 0&1\\1&0 \end{pmatrix} , \sqrt{-1} , \begin{pmatrix} 0&-1\\1&0 \end{pmatrix}\right).$$
Now by Corollary~\ref{cor_quat_self}, we can affirm that there are $3(11^3+2)=3999$ Euclidean self-orthogonal $Q_7$-codes over $\F_{11}$. Consequently, as aforementioned, there is an error in the formula provided at \cite[Theorem 4.5]{GaoYue21} which ensures the existence of $9$ codes with such property (see also \cite[Example~5.2~(2)]{GaoYue21}). Indeed, our computations with \textsc{Magma} together with Corollary~\ref{cor_quat_self} confirm that all the ideals of (\ref{eq:ex2}) with the following structure correspond to Euclidean self-orthogonal codes: $$ \pmb{0}\oplus \pmb{0}\oplus \pmb{0}\oplus \pmb{0}\oplus \left\langle \begin{pmatrix} 1&\lambda \\ 0&0 \end{pmatrix} \right\rangle, \; \text{ with } \lambda \in \F_{11^3}.$$ Observe that there are more than the $9$ claimed in \cite[Example~5.2~(2)]{GaoYue21}. The formula provided at \cite[Theorem 4.5]{GaoYue21} also seems to have been recently corrected in \cite{Cao2026} (see Theorem 3.5 and the final Remark of that paper).
\end{example}

\end{document}